\documentclass[letterpaper]{article} 
\usepackage{aaai24}  
\nocopyright
\usepackage{times}  
\usepackage{helvet}  
\usepackage{courier}  
\usepackage[hyphens]{url}  
\usepackage{graphicx} 
\usepackage{natbib}  
\usepackage{caption} 
\usepackage{algorithm}
\usepackage{algorithmic}

\usepackage{newfloat}
\usepackage{listings}
\DeclareCaptionStyle{ruled}{labelfont=normalfont,labelsep=colon,strut=off} 
\floatstyle{ruled}
\newfloat{listing}{tb}{lst}{}
\floatname{listing}{Listing}

\title{Omega-Regular Decision Processes\thanks{This work was supported in part by the EPSRC through grants EP/X017796/1 and EP/X03688X/1, the NSF through grant CCF-2009022 and the NSF CAREER award CCF-2146563; and the 
EU's Horizon 2020 research and innovation programme \includegraphics[height=8pt]{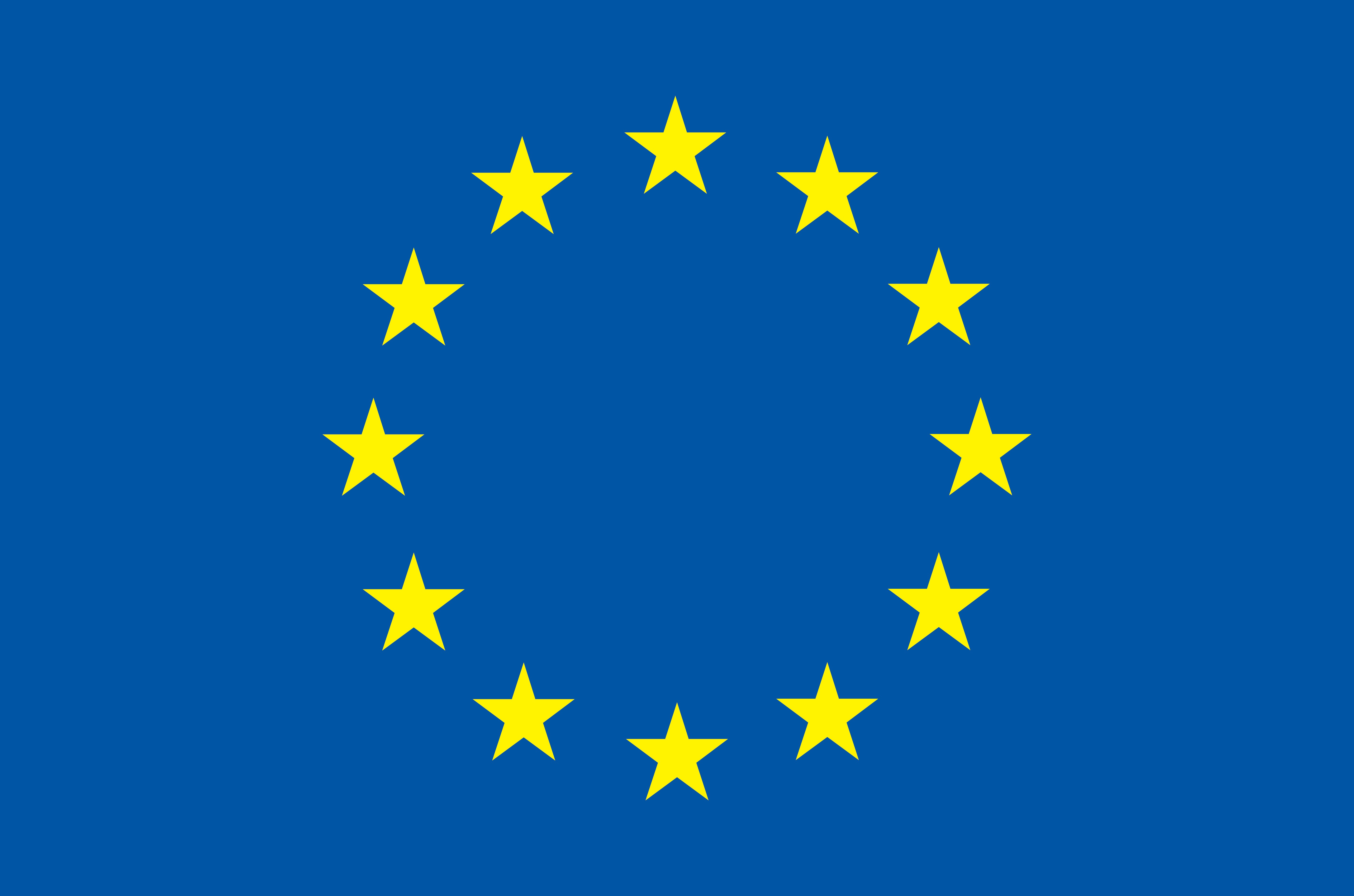}  under grant agreements No 864075 (CAESAR).}}
\author{
    E. M. Hahn\equalcontrib\textsuperscript{\rm 1}
    M. Perez\equalcontrib\textsuperscript{\rm 2},
    S. Schewe\equalcontrib\textsuperscript{\rm 3},
    F. Somenzi\equalcontrib\textsuperscript{\rm 2},
    A. Trivedi\equalcontrib\textsuperscript{\rm 2},
    D. Wojtczak\equalcontrib\textsuperscript{\rm 3}
}
\affiliations{
    \textsuperscript{\rm 1}University of Twente, The Netherlands\\
    \textsuperscript{\rm 2}University of Colorado Boulder, USA\\
    \textsuperscript{\rm 3}University of Liverpool, UK
}

\usepackage{lipsum,booktabs}

\usepackage{amssymb,amsmath,amsfonts,amsthm}
\usepackage{temporal}
\usepackage{fontawesome5} 
\usepackage{textcomp}
\usepackage[mathscr]{eucal}
\usepackage{svg}
\usepackage{tikz}

\usetikzlibrary{automata,positioning,decorations.markings,arrows,intersections,shapes,%
  shapes.symbols,calc,backgrounds,shadings}
\usetikzlibrary{overlay-beamer-styles}

\tikzset{
  >=stealth,
  box state/.style={draw,rectangle,rounded corners,fill=safecellcolor,minimum size=8mm},
  prob state/.style={draw,very thick,shape=circle,darkblue,minimum size=3mm,inner sep=0mm},
  node distance=2cm,on grid,auto, initial text=,
  every loop/.style={shorten >=0pt},
  accepting/.style={double distance=1.2pt, outer sep = 0.6pt+\pgflinewidth},
  accepting dot/.style={above=-2.5pt,circle,fill,darkgreen,inner sep=2pt,radius=1pt},
  loop above/.append style={every loop/.append style={out=120, in=60, looseness=6}},
  loop below/.append style={every loop/.append style={out=300, in=240, looseness=6}},
  loop left/.append style={every loop/.append style={out=210, in=150, looseness=6}},
  loop right/.append style={every loop/.append style={out=30, in=330, looseness=6}},
  accepting arc/.style={dashed},
  marked/.style={
    dashed,
    opacity=0.3
  },
  marked on/.style={alt=#1{marked}{}},
}

\usepackage{xcolor,colortbl}
\usepackage{mathtools}

\definecolor{darkgreen}{rgb}{0,0.6,0}
\definecolor{lightblue}{rgb}{0.5,0.6,1.0}
\definecolor{lightgray}{rgb}{0.98,0.98,0.98}
\definecolor{mauve}{rgb}{0.58,0,0.82}
\definecolor{sienna}{rgb}{0.6,0.18,0.09}
\colorlet{darkblue}{blue!60!black}
\colorlet{darkred}{red!50!black}
\colorlet{safecellcolor}{yellow!5}
\colorlet{goodcellcolor}{green!10}
\colorlet{badcellcolor}{blue!10}

\DeclareMathOperator{\infi}{inf}

\newcommand{\DIST}{{\mathcal D}}
\newcommand{\Ll}{\mathcal{L}}
\newcommand{\Aa}{\mathcal{A}}
\newcommand{\Mm}{\mathcal{M}}
\newcommand{\Nn}{\mathcal{N}}

\newcommand{\Ff}{\mathcal{F}}
\newcommand{\Rr}{\mathcal{R}}
\newcommand{\FRuns}{\mathit{FRuns}}
\newcommand{\Runs}{\mathit{Runs}}
\DeclareMathOperator{\last}{\mathit{last}}
\newcommand{\set}[1]{\left\{ #1 \right\}}

\newcommand{\seq}[1]{\langle #1 \rangle}
\DeclareMathOperator{\supp}{\mathit{supp}}
\newcommand{\eE}{\mathbb E}
\newcommand{\Real}{\mathbb R}

\DeclareMathOperator{\EDisct}{\mathsf{EDisct}}

\newcommand{\rank}{\mathsf{rank}}
\newcommand{\comp}{\mathsf{comp}}

\definecolor{Gray}{gray}{0.85}
\definecolor{LightCyan}{rgb}{0.88,1,1}

\def\run{\rho}

\newcommand{\Strat}{\Pi}
\newcommand{\vStrat}{\overline{\Pi}}
\def\rmdef{\stackrel{\mbox{\rm {\tiny def}}}{=}}
\newtheorem{theorem}{Theorem}
\newtheorem{proposition}{Proposition}

\newtheorem{corollary}{Corollary}
\newtheorem{lemma}{Lemma}
\newtheorem{definition}{Definition}

\newtheorem{example}{Example}

\begin{document}

\maketitle
\begin{abstract}
Regular decision processes (RDPs) are a subclass of non-Markovian decision processes where the transition and reward functions are guarded by some regular property of the past (a \emph{lookback}). 
While RDPs enable intuitive and succinct representation of non-Markovian decision processes, their expressive power coincides with finite-state Markov decision processes (MDPs).
We introduce omega-regular decision processes (ODPs) where the non-Markovian aspect of the transition and reward functions are extended to an $\omega$-regular lookahead over the system evolution.
Semantically, these lookaheads can be considered as \emph{promises} made by the decision maker or the learning agent about her future behavior.
In  particular, we assume that, if the promised lookaheads are not met, then the payoff to the decision maker is $\bot$ (least desirable payoff), overriding any rewards collected by the decision maker.
We enable optimization and learning for ODPs under the discounted-reward objective by reducing them to lexicographic optimization and learning over finite MDPs.
We present experimental results demonstrating the effectiveness of the proposed reduction.

\end{abstract}

\section{Introduction}
\label{sec:intro}
Markov decision processes (MDPs) are canonical models to express decision making under uncertainty, where the optimization objective is defined as a discounted sum of scalar rewards associated with various decisions. 
The optimal value and the optimal policies for MDPs can be computed efficiently via dynamic programming~\cite{Put94}. 
When the environment is not explicitly known but can be sampled in repeated interactions, reinforcement learning (RL)~\cite{Sutton18} algorithms combine stochastic approximation with dynamic programming to compute optimal values and policies.
RL, combined with deep learning~\cite{Goodfe16}, has emerged as a leading human-AI collaborative programming paradigm generating novel and creative solutions with ``superhuman'' efficiency~\cite{AlphaGo,wurman2022outracing,mirhoseini2020chip}.
A key shortcoming of this approach is the difficulty of translating designer's intent into a suitable reward signal.  
To help address this problem,
we extend MDPs with an expressive modeling primitive---called \emph{promises}---that improves the communication between the learning agent and the programmer. 
We dub these processes \emph{$\omega$-regular decision processes} (ODPs).

\paragraph{Motivation.}
A key challenge in posing a decision problem as an MDP is to define a scalar \emph{reward signal} that is Markovian (history-independent) on the state space. 
While some problems, such as reachability and safety, naturally lend themselves to a reward-based formulation, such an interface is often cumbersome and arguably error-prone. 
This difficulty has been well documented, especially within the RL literature, under different terms including \emph{misaligned specification}, \emph{specification gaming}, and \emph{reward hacking}~\cite{pan2022effects,amodei2016concrete,yuan2019novel,skalse2022defining,Coastrunners}.

To overcome this challenge, automata and logic-based reward gadgets---such as reward machines, $\omega$-regular languages, and LTL---have been proposed to extend the MDP in the context of planning~\cite{Baier08} and, more recently, of RL~\cite{RM1,camacho2019ltl,Sadigh14,Hahn19,Fu14}.
In these works, an interpreter provides a reward for the actions of the decision maker by monitoring the action sequences with the help of the underlying reward gadget.
While such reward interface is convenient from the programmer's perspective, it limits the agency of the decision maker in claiming rewards for her actions by making it opaque.

Brafman \emph{et al.}~\cite{Brafm19} initiated a formal study of non-Markovian MDPs in the planning setting, and proposed regular decision processes (RDPs) as a tractable representation of non-Markovian MDPs.
Abadi \emph{et al.}~\cite{Abadi21} extended this work by combining Mealy machine learning with RL. 
In an RDP, the agent can choose a given action and collect its associated reward as long as the partial episode satisfies a certain regular property provided as the \emph{guard} to that action.
This modeling feature allows/expects the agent to keep in memory some regular information about the past in order to choose her actions optimally. 
Appending the MDPs with retrospective memory allows for a succinct and transparent modeling. On the other hand, adding memory as a regular language does not increase the expressive power of MDPs and RDPs can be compiled into finite MDPs~\cite{Abadi21} recovering the tractability of optimization and learning.

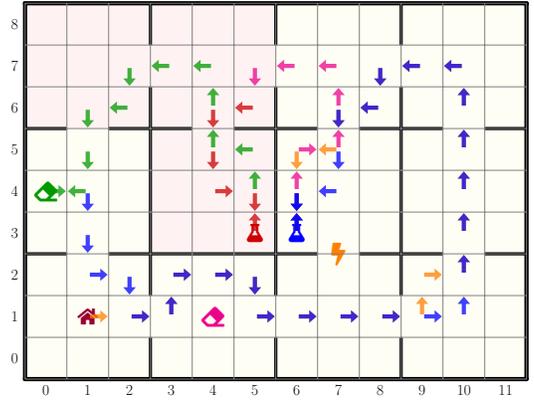
\begin{figure*}[t]
  \centering
  \begin{minipage}[c]{0.58\linewidth} 
  \begin{tabular}{c|c|c|c}
    where & LTL formula & type & reward \\\hline
    \textcolor{purple!80!black}{\faHome} & $(\always\eventually \mathtt{clean\_lab}) \wedge (\always\eventually \mathtt{dirty\_lab})$ & promise & $0$ \\
    \textcolor{purple!80!black}{\faHome} & $\eventually\always \neg \mathtt{initial\_location}$ & promise & $-\xi$ \\
    \textcolor{red!80!black}{\faFlask} & $\neg \mathtt{dirty\_lab} \since \mathtt{clean\_lab}$ & guard & $\rho$ \\
    \tikz\draw[fill=red!8] (0,0) rectangle (0.3,0.3); & $\mathtt{decontamination} \releases \neg \mathtt{clean\_lab}$ & promise & $0$ \\
    \textcolor{green!60!black}{\faEraser} & $\mathtt{true}$ & & $-f_1$ \\
    \textcolor{magenta}{\faEraser} & $\mathtt{true}$ & & $-f_2$
  \end{tabular}
  \end{minipage}
  \begin{minipage}[c]{0.4\linewidth}
      \tikzset{
  strat/.pic={
    \fill [fill opacity=0.75,rotate around={#1:(0,0)}] (0.1,-0.1) -- (0.1,0.1) -- (0.6,0.1) --
    (0.6,0.3) -- (0.95,0) -- (0.6,-0.3) -- (0.6,-0.1) --cycle;
  }
}

  \begin{tikzpicture}[scale=0.2777777777777778,transform shape]
    \edef\e{0}
    \edef\n{90}
    \edef\w{180}
    \edef\s{270}

    \fill[yellow!5] (0,0) rectangle (24,18);
    \fill[red!5] (6,6) -- ++(6,0) -- ++(0,12) -- ++(-12,0) -- ++(0,-6) -- ++(6,0)-- cycle;
    \foreach \y in {0,...,11} {
      \path (1+2*\y,-0.5) node {\Huge$\y$};
    }
    \draw[very thick] (0,6) -- ++(2,0) ++(2,0) -- ++(10,0) ++(2,0) -- ++(4,0) ++(2,0) -- ++(2,0);
    \draw[very thick] (0,12) -- ++(2,0) ++(2,0) -- ++(4,0) ++ (2,0) -- ++(4,0) ++(2,0) -- ++(4,0) ++(2,0) -- ++(2,0);
    \foreach \y in {6,12,18} {
      \draw[very thick] (\y,0) -- ++(0,2) ++(0,2) -- ++(0,10) ++(0,2) -- ++(0,2);
    }
    \foreach \y in {0,...,8} {
      \path (-0.5,1+2*\y) node {\Huge$\y$};
    }
    \node[scale=1.2] (home)    at (3,3)  {\Huge\textcolor{purple!80!black}{\faHome}};
    \node[scale=1.2] (wrench)  at (13,7) {\Huge\textcolor{blue}{\faFlask}};
    \node[scale=1.2] (flask)   at (11,7) {\Huge\textcolor{red!80!black}{\faFlask}};
    \node[scale=1.5] (eraser1) at (1,9)  {\Huge\textcolor{green!60!black}{\faEraser}};
    \node[scale=1.5] (eraser2) at (9,3)  {\Huge\textcolor{magenta}{\faEraser}};
    \node[scale=1.5] (bolt)    at (15,6) {\Huge\textcolor{orange}{\faBolt}};
    \draw[xstep=24cm,ystep=18cm,ultra thick,cap=round] (0,0) grid (24,18);
    \draw[step=2cm,ultra thin,gray] (0,0) grid (24,18);
    \pic[orange] at (3,3) {strat=\e};
    \pic[orange] at (5,3) {strat=\e};
    \pic[orange] at (7,3) {strat=\n};
    \pic[orange] at (7,5) {strat=\e};
    \pic[orange] at (9,5) {strat=\e};
    \pic[orange] at (11,5) {strat=\s};
    \pic[orange] at (11,3) {strat=\e};
    \pic[orange] at (13,3) {strat=\e};
    \pic[orange] at (15,3) {strat=\e};
    \pic[orange] at (17,3) {strat=\e};
    \pic[orange] at (19,3) {strat=\n};
    \pic[orange] at (19,5) {strat=\e};
    \pic[orange] at (21,5) {strat=\n};
    \pic[orange] at (21,7) {strat=\n};
    \pic[orange] at (21,9) {strat=\n};
    \pic[orange] at (21,11) {strat=\n};
    \pic[orange] at (21,13) {strat=\n};
    \pic[orange] at (21,15) {strat=\w};
    \pic[orange] at (19,15) {strat=\w};
    \pic[orange] at (17,15) {strat=\s};
    \pic[orange] at (17,13) {strat=\w};
    \pic[orange] at (15,13) {strat=\s};
    \pic[orange] at (15,11) {strat=\w};
    \pic[orange] at (13,11) {strat=\s};
    \pic[orange] at (13,9) {strat=\s};
    \pic[orange] at (13,7) {strat=\n};
    \pic[blue] at (13,9) {strat=\s};
    \pic[blue] at (13,7) {strat=\n};
    \pic[magenta] at (13,9) {strat=\n};
    \pic[magenta] at (13,11) {strat=\e};
    \pic[magenta] at (15,11) {strat=\n};
    \pic[magenta] at (15,13) {strat=\n};
    \pic[magenta] at (15,15) {strat=\w};
    \pic[magenta] at (13,15) {strat=\w};
    \pic[magenta] at (11,15) {strat=\s};
    \pic[red!80!black] at (11,13) {strat=\w};
    \pic[red!80!black] at (9,13) {strat=\s};
    \pic[red!80!black] at (9,11) {strat=\s};
    \pic[red!80!black] at (9,9) {strat=\e};
    \pic[red!80!black] at (11,9) {strat=\s};
    \pic[red!80!black] at (11,7) {strat=\n};
    \pic[green!60!black] at (11,9) {strat=\n};
    \pic[green!60!black] at (11,11) {strat=\w};
    \pic[green!60!black] at (9,11) {strat=\n};
    \pic[green!60!black] at (9,13) {strat=\n};
    \pic[green!60!black] at (9,15) {strat=\w};
    \pic[green!60!black] at (7,15) {strat=\w};
    \pic[green!60!black] at (5,15) {strat=\s};
    \pic[green!60!black] at (5,13) {strat=\w};
    \pic[green!60!black] at (3,13) {strat=\s};
    \pic[green!60!black] at (3,11) {strat=\s};
    \pic[green!60!black] at (3,9) {strat=\w};
    \pic[green!60!black] at (1,9) {strat=\e};
    \pic[blue] at (3,9) {strat=\s};
    \pic[blue] at (3,7) {strat=\s};
    \pic[blue] at (3,5) {strat=\e};
    \pic[blue] at (5,5) {strat=\s};
    \pic[blue] at (5,3) {strat=\e};
    \pic[blue] at (7,3) {strat=\n};
    \pic[blue] at (7,5) {strat=\e};
    \pic[blue] at (9,5) {strat=\e};
    \pic[blue] at (11,5) {strat=\s};
    \pic[blue] at (11,3) {strat=\e};
    \pic[blue] at (13,3) {strat=\e};
    \pic[blue] at (15,3) {strat=\e};
    \pic[blue] at (17,3) {strat=\e};
    \pic[blue] at (19,3) {strat=\e};
    \pic[blue] at (21,3) {strat=\n};
    \pic[blue] at (21,5) {strat=\n};
    \pic[blue] at (21,7) {strat=\n};
    \pic[blue] at (21,9) {strat=\n};
    \pic[blue] at (21,11) {strat=\n};
    \pic[blue] at (21,13) {strat=\n};
    \pic[blue] at (21,15) {strat=\w};
    \pic[blue] at (19,15) {strat=\w};
    \pic[blue] at (17,15) {strat=\s};
    \pic[blue] at (17,13) {strat=\w};
    \pic[blue] at (15,13) {strat=\s};
    \pic[blue] at (15,11) {strat=\s};
    \pic[blue] at (15,9) {strat=\w};
    \pic[blue] at (13,9) {strat=\s};
    \pic[blue] at (13,7) {strat=\n};
  \end{tikzpicture}
  \end{minipage}
  \caption{A grid-world model of a biological lab with clean and dirty areas. The strategy shown here is computed by RL based on the method proposed in this paper.  
    The rewards satisfy $\xi > 0$ and $0 < f_1 < f_2 < \rho$.  Promises and guards are specified in LTL \cite{Pnueli77}, with future and past operators.  $\releases$ denotes the \emph{releases} operator: decontamination removes the constraint $\neg \mathtt{clean\_lab}$.  $\since$ denotes the \emph{since} operator: the robot comes from the clean lab and has not been to the dirty lab since.}
  \label{fig:biolab}
\end{figure*}
As a dual capability to the retrospective memory, we propose extending the RDP framework with the ``prospective memory''~\cite{mcdaniel2007prospective} (also known as memory for intentions) to allow the agent to make promises about the future behavior and collect rewards based on this promise. 
We posit that such an abstraction will allow the agent to declare her intent to the environment and collect reward, and will result in more explainable and transparent behavior. 
This is the departure point for $\omega$-regular decision processes, which we now introduce with the help of the following example.
We note that while this example is little busy, it showcases multiple features of our framework.

\begin{example}
Consider the grid world shown in Fig.~\ref{fig:biolab}, where a robot has to repeatedly visit two labs, one ``clean'' (blue) and one ``dirty'' (red).  Whenever the robot passes through the dirty area---highlighted with a rose background---it has to visit a decontamination station (in one of the two cells marked with an eraser) before it can re-enter the clean lab.  Every time the robot visits the dirty lab, it collects a reward if it just arrived from the clean lab.

The two decontamination stations charge different fees.  The cheaper one requires a detour from the shortest route.  Both charge less than the robot earns by visiting the two labs.  The clean lab has two doors.  The one on the south side, however, is equipped with a ``zapper'' that has to be disabled on first crossing.  If the robot manages to disable the zapper, it secures a shorter route and collects rewards more often; if it fails, it cannot complete its task.  If the probability that the robot is put out of commission is sufficiently low, then a strategy that maximizes the expected cumulative reward will try to disable the zapper, while a strategy that maximizes the probability of carrying out the task will choose the longer, safer route.  Let us assume the latter is desired. Finally, let us also assume that the robot should not re-enter its initial location more than a finite number of times.  Figure~\ref{fig:biolab} summarizes the specifications and details how they are expressed as rewards and promises.  In this case, promises are associated to states; i.e., to all transitions emanating from the designated states.  No lookbacks are necessary, though the promise made in the dirty area could be turned into a guard on the entrance to the clean lab.

The combination of $\omega$-regular properties and rewards makes for a flexible and natural way to describe the objective of the decision maker.  There may seem to be redundancy in the specification: why rewarding the robot for visiting the labs if it is already forced to visit them by the $\always\eventually$ requirements?  However, a proper combination of $\omega$-regular and quantitative specifications may give strategies that simultaneously optimize short-term (discounted) reward and guarantee satisfaction of long-term goals (when such strategies exist).  Without the $\omega$-regular requirement, the robot of Fig.~\ref{fig:biolab} would try its luck with the zapper.  Without the reward collected on each visit to the dirty lab, the robot would only have $\varepsilon$-optimal strategies, which would postpone satisfaction of the $\omega$-regular part of the specification to avoid the decontamination fees.  Such postponement strategies are seldom practically satisfactory.  Formulating the problem as an $\omega$-regular decision process helps one prevent their occurrence.
The strategy shown in Figure~\ref{fig:biolab} is computed using RL
based on the techniques presented in this paper.
\end{example}

\paragraph{Contributions.} We introduce $\omega$-regular decision processes (ODPs) that generalize regular decision processes with prospective memory (promises) modeled as $\omega$-regular lookaheads. 
We show decidability (Theorem~\ref{theo:hard}) of the optimal discounted reward optimization problem for ODPs. In particular, we show that computing $\varepsilon$-optimal strategies is: 1) EXPTIME-hard when the lookaheads are given as universal co-B\"uchi automata (UCW) and 2) 2EXPTIME-hard when they are expressed in LTL.

A key construction of the paper is the translation of the lookaheads to a lexicographic optimization problem over MDPs. This construction creates a nondeterministic B\"uchi automaton (NBA) to test whether all promises made are almost surely fulfilled. This procedure involves a complementation procedure from UCAs to NBA. 
    To be able to use this reduction for model checking or reinforcement learning, a critical requirement is to design an NBA that is \emph{good-for-MDP} (GFM)~\cite{Hahn20}.
    We provide a rank based complementation construction to demonstrate that the resulting automata are GFM.   
    We also show that leading rank-based complementation procedures all deliver good-for-MDP automata, enabling off-the-shelf complementation constructions to be used for OMDPs.

    We have also implemented the proposed construction to remove $\omega$-regular lookaheads from the MDPs. To demonstrate the experimental performance of our reduction, we present experiments on randomly generated examples.

\section{Preliminaries} 
\label{sec:background}

\subsection{Markov Decision Processes}
Let $\DIST(S)$ denote the set of all discrete distributions over $S$.
 A Markov decision process (MDP) $\Mm$ is a tuple $(S, s_0, A, T, AP, L)$ where $S$ is a finite set of states, $s_0 \in S$ is the initial state, $A$ is a finite set of {\it actions},  $T\colon S \times A \to \DIST(S)$ is the {probabilistic transition function}, $AP$ is the set of {\it atomic propositions} (observations), and $L\colon S \to 2^{AP}$ is the {\it labeling function}.

For any state $s \in S$, we let $A(s)$ denote the set of actions that can be selected in
state $s$.
An MDP is a Markov chain if $A(s)$ is singleton for all $s \in S$.
For states $s, s' \in S$ and $a \in A(s)$, $T(s,
a)(s')$ equals $\Pr (s' | s, a)$.
A {\it run} of $\Mm$ is an $\omega$-word $\seq{s_0, a_1, s_1, \ldots} \in S
\times (A \times S)^\omega$ such that $\Pr(s_{i+1} | s_{i}, a_{i+1}) {>} 0$ for all $i
\geq 0$.
A finite run is a finite such sequence. 
We write $\Runs^\Mm (\FRuns^\Mm)$  for the set of runs (finite runs) of the MDP
$\Mm$  and $\Runs{}^\Mm(s) (\FRuns{}^\Mm(s))$  for the set of runs (finite runs) of
the MDP $\Mm$ starting from the state $s$.  We write $\last(r)$ for the last state
of finite run $r$.

We write $\Sigma \rmdef 2^{AP}$ for the alphabet of the set of labels. 
For a {\it run} $r = \seq{s_0, a_1, s_1, \ldots}$ we define the corresponding
labeled run as $L(r) = \seq{L(s_0), L(s_1), \ldots} \in (\Sigma)^\omega$.

\paragraph{Strategies.} A {\it strategy} in $\Mm$ is a function $\sigma \colon \FRuns \to \DIST(A)$ such that
$\supp(\sigma(r)) \subseteq A(\last(r))$, where $\supp(d)$ denotes the support
of the distribution $d$.
A strategy $\sigma$ is {\it pure} if $\sigma(r)$ is a point
distribution for  all runs $r \in \FRuns^\Mm$ and is {\it mixed} if $\supp(\sigma(r)) = A(\last(r))$ for  all runs
$r \in \FRuns^\Mm$.
Let $\Runs^\Mm_\sigma(s)$ denote the subset of runs $\Runs^\Mm(s)$ that
correspond to strategy $\sigma$ with initial state $s$.
Let $\Strat_\Mm$ be the set of all strategies.
We say that $\sigma$ is {\it stationary} if $\last(r) = \last(r')$ implies
$\sigma(r) = \sigma(r')$ for all finite runs $r, r' \in \FRuns^\Mm$.
A stationary strategy can be given as a function $\sigma: S \to \DIST(A)$.  
A strategy is {\it positional} if it is both pure and stationary.

\paragraph{Probability Space.} An MDP $\Mm$ under a strategy $\sigma$ results in a Markov chain $\Mm_\sigma$.
If $\sigma$ is finite memory, then $\Mm_\sigma$ is a finite-state
Markov chain.
The behavior of $\Mm$ under a strategy $\sigma$ from
$s \in S$ is defined on the probability space
$(\Runs^\Mm_\sigma(s), \Ff_{\Runs^\Mm_\sigma(s)}, \Pr^\Mm_\sigma(s))$ over
the set of infinite runs of $\sigma$ with starting state $s$.  Given a random variable 
$f \colon \Runs^\Mm \to \Real$, we denote by $\eE^\Mm_\sigma(s) \set{f}$ the
expectation of $f$ over the runs of $\Mm$ originating at $s$ that
follow $\sigma$.

\paragraph{Reward Machines.}
The learning objective over MDPs in RL is often expressed using a Markovian reward function, i.e. a function $\rho\colon S \times A \times S \to \Real$ assigning utility to transitions.
A {\it rewardful} MDP is a tuple $\Mm = (S, s_0, A, T, \rho)$ where $S, s_0, A,$ and $T$ are defined as for MDP, and $\rho$ is a Markovian reward function.
A rewardful MDP $\Mm$ under a
strategy $\sigma$ determines a sequence of random rewards
${\rho(X_{i-1}, Y_i, X_i)}_{i \geq 1}$, where $X_i$ and $Y_i$ are the
random variables denoting the $i$-th state and action, respectively.
For $\lambda \in [0, 1[$, the {\it discounted reward} $\EDisct(\lambda)^\Mm_\sigma(s)$ from a state $s \in S$ under strategy $\sigma$ is defined as
\begin{equation}
\lim_{N \to \infty} \eE^\Mm_\sigma(s) \Big\{\sum_{1 \leq i \leq N}
  \lambda^{i-1} \rho(X_{i-1}, Y_i, X_i)\Big\}.\label{eq1}
\end{equation}

We define the optimal discounted reward
$\EDisct^\Mm_*(s)$ for a state $s \in S$ as 
$\EDisct^\Mm_*(s) \rmdef \sup_{\sigma \in \Strat_\Mm} \EDisct^\Mm_\sigma(s)$.
A strategy $\sigma$ is discount-optimal if
$\EDisct^\Mm_\sigma(s) = \EDisct^\Mm_*(s)$ for all $s {\in} S$.
The optimal discounted cost can be computed in polynomial time~\cite{Put94}.  

Often, complex learning objectives cannot be expressed using Markovian reward signals. 
A recent trend is to resort to finite-state reward
machines~\cite{icarte2022reward}. 
A reward machine is a tuple $\Rr = (\Sigma, U, u_0, \delta, \rho)$
where $U$ is a finite set of states, $u_0 {\in} U$ is the starting state,
$\delta\colon U {\times} \Sigma \to 2^U$ is the transition function, 
and $\rho\colon U \times \Sigma \times U \to \Real$ is the reward function.
Given an MDP $\Mm = (S, s_0, A, T, AP, L)$ and a reward machine $\Rr = (2^{AP}, U, u_0, \delta, \rho)$,  their product 
$\Mm{\times}\Rr = (S{\times} U, (s_0,u_0), (A {\times} U),
T^\times, \rho^\times)$
is a rewardful MDP where the transition function 
$T^\times ((s,u), (a, u'))(({s}',{u}'))$ equals $T(s,a)({s}')$ if $u' {\in} \delta(u,L(s))$ and equals $0$ otherwise. Moreover, the reward function $\rho^\times((s,u), (a, u'), (s', u')) $ equals $\rho(u, L(s), u')$ if $(u,L(s),{u}') \in \delta$ and is $0$ otherwise.
For discounted objectives, the optimal strategies of
$\Mm{\times}\Rr$ are positional on $\Mm{\times}\Rr$, inducing 
finite memory strategies 
over $\Mm$ maximizing the learning objective given by $\Rr$. 

\subsection{Omega-Regular Languages}

A deterministic finite state automaton (DFA) is a tuple
${\mathcal A} = (\Sigma,Q,q_0,\delta,F)$, where $\Sigma$ is a finite
\emph{alphabet}, $Q$ is a finite set of \emph{states}, $\delta \colon Q \times \Sigma \to 2^Q$ is the \emph{transition function}, and $F \subset Q$ is the set of \emph{accepting (final) states}.
A \emph{run} $r$ of ${\mathcal A}$ on $w = w_0\ldots w_{n-1} \in \Sigma^*$ from an initial state $q_0 \in Q$
is a finite word $r_0, w_0, r_1, w_1, \ldots, r_n$ in
$Q \times (\Sigma \times Q)^*$ such that $r_0 = q_0$ and, for $0 < i \leq  n$, $r_i \in \delta(r_{i-1},w_{i-1})$.  
We write $\last(r)$ for the last state of the finite run $r$.
A run $r$ of ${\mathcal A}$ is \emph{accepting} if $\last(r) \in
F$. 
The \emph{language} $\Ll(\Aa, q)$ of ${\mathcal A}$ is the set of words in $\Sigma^*$ with accepting runs in ${\mathcal A}$ from $q$.

\paragraph{$\omega$-Automata.}
A (nondeterministic) \emph{B\"uchi automaton}  (NBA) is a tuple
${\mathcal A} = (\Sigma, Q, q_0, \delta, \gamma)$, where $\Sigma$ is a finite
\emph{alphabet}, $Q$ is a finite set of \emph{states}, $\delta \colon Q \times \Sigma \to 2^Q$ is the \emph{transition function}, and
$\gamma \colon Q \times \Sigma \to 2^Q$ with $\gamma(q,\sigma) \subseteq \delta(q,\sigma)$ for all $(q,\sigma) \in Q \times \Sigma$ are the \emph{accepting transitions}.
A \emph{run} $\run$ of ${\mathcal A}$ on $w \in \Sigma^\omega$ from the initial state $q_0 \in Q$ is an $\omega$-word $\run_0, w_0, \run_1, w_1, \ldots$ in
$(Q \times \Sigma)^\omega$ such that $\run_0 = q_0$ and, for all $i > 0$,
$\run_i \in \delta(\run_{i-1},w_{i-1})$.  
We write $\infi(\run)$ for the set of transitions that appear infinitely
often in the run $\run$.
A run $\run$ of an NBA ${\mathcal A}$ is \emph{accepting} if $\infi(\run)$ contains a transition from $\gamma$.
The \emph{language} $\Ll(\Aa, q)$ of ${\mathcal A}$ is the subset of words in $\Sigma^\omega$ that have accepting runs in ${\mathcal A}$ from $q$.
A language is $\omega$-\emph{regular} if it is accepted by a nondeterministic B\"uchi automaton.

A universal co-B\"uchi automaton (UCA) ${\mathcal A} = (\Sigma,Q,q_0, \delta,\gamma)$ is the dual of an NBA and its language can be defined using the notion of \emph{rejecting runs}. 
We call a transition in $\gamma$ \emph{rejecting} and any runs with a transition in $\gamma$ occurring infinitely often \emph{rejecting} runs.
The language $\Ll(\Aa, q)$ of a UCA ${\mathcal A}$ is the set of $\omega$-words starting from $q$ that do not have a rejecting run.
A UCA therefore recognizes the complement of a structurally identical NBA.

\paragraph{Good-for-MDP Automata.}
Given an MDP $\Mm$ and a NBA automaton $\Aa$, the probabilistic model checking problem is to find a strategy that maximizes the probability of generating words in the language of $\Aa$.
Automata-theoretic tools provide an algorithm for probabilistic model checking when the NBA satisfies the so-called \emph{good-for-MDP} property~\cite{Hahn20}.
An NBA $\mathcal A$ is called \emph{good-for-MDPs} if, for any MDP $\Mm$, controlling $\Mm$ to maximize the chance that its trace is in the language of $\mathcal A$ and controlling the syntactic product $\Mm {\times}\Aa$ (defined next) to maximize the chance of satisfying the B\"uchi objective are the same. In other words, for any given MDP, the nondeterminism of $\mathcal A$ can be resolved on-the-fly.

Given an MDP $\Mm = ( S, s_0, A, T, AP, L )$
and an (UCA or NBA) automaton $\mathcal{A} = (2^{AP}, Q, q_0, \delta,\gamma)$,
their \emph{product}
$\Mm \times \mathcal{A} = ( S {\times} Q, (s_0,q_0), A {\times} Q, T^\times, F^\times )$ is an MDP where the transition function 
$T^\times ((s,q),(a,q'))(({s}',{q}'))$ equals $T(s,a)({s}')$ if $(q,L(s,a,{s}'),{q}') {\in} \delta$ and it is $0$ otherwise.
The set of accepting transitions in the case of NBA or rejecting transitions in the case of UCA,
$F^\times \subseteq (S \times Q) \times (A \times Q) \times (S
\times Q)$, is defined by $((s,q),(a,q'),(s',q')) \in F^\times$
iff $(q,L(s,a,s'),q') \in F$ and $T(s,a)(s') > 0$.
A strategy $\sigma$ on the product induces a strategy $\sigma'$ on the MDP with the same value, and vice versa.  
Note that for a stationary
$\sigma$ on the product, the strategy $\sigma'$ on the MDP needs memory.

An {\it end-component} of an MDP $\Mm$ is a sub-MDP $\Mm'$ s.t. for every state pair $(s, s')$ in $\Mm'$ there is a strategy to reach $s'$ from $s$ with positive probability. 
A maximal end-component is an end-component that is maximal under
set-inclusion.
An accepting/rejecting end-component is an end-component that contains an accepting/rejecting transition.

\section{Omega-Regular Decision Processes}
\label{sec:omdp}
The Regular decision processes (RDPs)~\cite{Abadi21} depart from the Markovian assumption of MDPs by allowing transitions and reward functions to be {\it guarded} (retrospective memory) by a regular property of the history. 
To build on this idea, we propose $\omega$-regular decision processes (ODPs), where transitions and rewards are not only constrained by regular properties on the history but where the decision maker may also make {\it promises} (prospective memory) to limit their future choices in exchange for a better reward or evolution.
ODPs offer a convenient framework for non-Markovian systems by allowing the decision maker to combine $\omega$-regular objectives and scalar rewards.

For an automaton of any type, an \emph{automaton schema} $\mathcal A = (\Sigma,Q,\delta,F)$ (for DFA) or $\mathcal A = (\Sigma,Q,\delta,\gamma)$ (for NBAs or UCAs) is defined as an automaton without an initial state.
For an automaton schema $\mathcal A = (\Sigma,Q,\delta,\gamma)$ and a state $q \in Q$, we write $\mathcal A_q = (\Sigma,Q,q,\delta,\gamma)$ 
as the automaton with $q$ as initial state and 
$\Ll(\Aa, q)$ for its language.
We express various transition guards using a DFA schema (lookback automaton) and various promises using a UCA schema\footnote{\noindent\textbf{Why UCAs?} We opted for the use of UCAs, instead of NBAs, in our ODP framework due to the accumulation of promises during a run of an ODP. 
As new promises are made, previous promises must also be satisfied, leading to a straightforward operation on UCAs. However, this same operation on NBAs would result in alternating automata, adding an additional exponential blow-up to our construction.
UCAs are becoming increasingly prevalent in both the formal methods~\cite{finkbeiner2013bounded,filiot2009antichain,dimitrova2018maximum} and AI~\cite{camacho2018ltl,camacho2019strong} communities. They are often referred to as NBAs that recognize the complement language. It is worth noting that if an NBA $\Aa$ recognizes the models of an LTL or QPTL formula $\phi$, or any other specification logic with negation, then $\Aa$, read as a UCA, recognizes $\neg \phi$ and vice versa. Therefore, the same automata translations can be applied to these specification languages.} (lookahead automaton).

\begin{definition}[Omega-Regular Decision Processes]
An $\omega$-regular decision process (ODP) $\Mm$ is a tuple $(S, s_0, A, T, r, \Aa_a, \Aa_b, AP, L)$ where:
\begin{itemize}
    \item $S$ is a finite set of states,
    \item $s_0 \in S$ is the initial state, 
    \item $A$ is a finite set of {\it actions}, 
    \item $AP$ is the set of {\it atomic propositions},
    \item $L: S \to 2^{AP}$ is the {\it labeling function},
    \item $\Aa_b = (2^{AP}, Q_b, \delta_b, F_b)$ is a {\it lookback} DFA schema,
    \item $\Aa_a = (2^{AP}, Q_a, \delta_a, \gamma_a)$ is the {\it lookahead} UCA schema,  
    \item $T : S {\times} Q_b {\times} A {\times} Q_a {\to} \DIST(S)$ is the {transition function},  
    \item and $r: S {\times} Q_b {\times} A {\times} Q_a {\to} \Real$ is the {\it reward function}. 
\end{itemize}
An ODP with trivial lookahead $\Ll(\Aa_a, q) = \Sigma^\omega$, for every $q \in Q_a$, is a {\em regular decision process (RDP)}.
An ODP with trivial lookback $\Ll(\Aa_b, q) = \Sigma^*$, for every $q \in Q_b$, is a  lookahead decision process (LDP).
An ODP with trivial lookahead and lookback is simply an MDP. 
In these special cases, we will omit the trivial language from its description. 
\end{definition}

A {\it run} $\seq{s_0, (\beta_1, a_1, \alpha_1), s_1, (\beta_2, a_2, \alpha_2), \ldots} \in S
\times ((Q_b \times A \times Q_a) \times S)^\omega$ of $\Mm$ is an $\omega$-word such that $\Pr(s_{i+1} | s_{i}, (\beta_{i+1}, a_{i+1}, \alpha_{i+1})) {>} 0$ for all $i
\geq 0$.
A finite run is a finite such sequence. 
We say that a run $\seq{s_0, (\beta_1, a_1, \alpha_1), s_1, (\beta_2, a_2, \alpha_2), \ldots} \in S
\times ((Q_b \times A \times Q_a) \times S)^\omega$ is a {\it valid} run if for every $i \geq 1$ we have that $L(s_0)L(s_1)\cdots L(s_{i-1}) \in \Ll(\Aa_b, \beta_i)$ and 
$L(s_i)L(s_{i+1})\cdots \in \Ll(\Aa_a, \alpha_i)$.
The concepts of strategies, memory, and probability space are defined for the ODPs in an analogous manner to MDPs.
We say that a strategy $\sigma$ for an ODP is a \emph{valid} strategy if the resulting runs are almost surely valid.
Let $\vStrat_\Mm$ be the set of all valid strategies of $\Mm$.

The expected discounted reward $\EDisct(\lambda)^\Mm_\sigma(s)$ for a strategy in an ODP $\Mm$ is defined as in (\ref{eq1}).
We define the optimal discounted reward
$\EDisct^\Mm_*(s)$ for a state $s \in S$ as 
$\EDisct^\Mm_*(s) \rmdef \sup_{\sigma \in \vStrat_\Mm} \EDisct^\Mm_\sigma(s)$.
A strategy $\sigma$ is discount-optimal if
$\EDisct^\Mm_\sigma(s) = \EDisct^\Mm_*(s)$ for all $s {\in} S$.
Given $\varepsilon > 0$, we say that a strategy $\sigma$ is $\varepsilon$-optimal if
$\EDisct^\Mm_\sigma(s) \geq \EDisct^\Mm_*(s) - \varepsilon$ for all $s {\in} S$.
The key optimization problem for ODPs is to compute the optimal discounted reward and a discount-optimal strategy.
However, such strategy may not always exist as shown next. 
\begin{example}
Consider an ODP where one can freely choose the next letter from the alphabet $\{a,b\}$ and have a reward of $1$ for $a$ and $0$ for $b$. 
With each transition the lookback is trivial $\Ll(\Aa_b, q) = \Sigma^*$ and the lookahead is $\Sigma^*(b\Sigma^*)^\omega$ (infinitely many $b$'s).
While we cannot achieve the discounted reward of $\sum_{i=0}^\infty \lambda^i = \frac{1}{1-\lambda}$ with any valid strategy, we can get arbitrarily close to this value by, e.g., choosing $a$'s until a reward $>\frac{1}{1-\lambda}{-}\varepsilon$ is collected for any given $\varepsilon {>} 0$, and henceforth choose $b$'s.
While the optimal B\"uchi-discounted value is $\frac{1}{1-\lambda}$, no (finite or infinite memory) strategy can attain this value.
\end{example}

Throughout the rest of this paper, we will focus on the problem of computing optimal discounted values and $\varepsilon$-optimal strategies for ODPs. However, before we dive into the general problem, it is helpful to examine some important subclasses of ODPs.
\begin{theorem}[Removing Lookbacks~\cite{Abadi21}]
\label{theo:noLookingBack}
For any given RDP $\Mm = (S, s_0, A, T, r, \Aa_b)$, we can construct an MDP $\Nn = (S', s_0', A', T', r')$ such that the optimal discounted value starting from $s_0$ in $\Mm$, denoted by $\EDisct^{\Mm}_*(s_0)$, is equal to the optimal discounted value starting from $s_0'$ in $\Nn$, denoted by $\EDisct^\Nn_*(s_0')$. A finite-memory optimal strategy for $\Mm$ can be computed from an optimal strategy for $\Nn$.
\end{theorem}
\begin{proof}
Simulating the lookback automaton $\Aa_b$ is a straightforward process. Without loss of generality, we can assume that $(\Aa_b, p)$ is deterministic for all $p \in Q_b$. We can simulate $\Aa_b$ by computing, for each state $p\in Q_b$, the state $\alpha(p)\in Q_b$ that has been reached so far by $(\Aa_b, p)$ on the current prefix (if it exists; otherwise, $\alpha(p)$ is undefined).
A transition of $(S, s_0, A, T, r, \Aa_b)$ with a lookback $r\in Q_b$ can be triggered whenever $F_b \cap \alpha(r)\neq\emptyset$.
\end{proof}

Moving forward, we will assume that the ODP we are working with has a trivial lookback.

\paragraph{Complexity.}
It is easy to see that the optimization problem for ODPs is EXPTIME-hard, even for lookahead MDPs.
This is due to the special case where the initial state of a lookahead MDP has no incoming transitions,
and we can assign a payoff of $1$ for the promise to satisfy a property given by a UCA and a $0$ reward in all other cases.
The problem then reduces to checking if the MDP can be controlled to create a word in the language of the UCA (or a model of the LTL formula) almost surely. 
If the specification can be satisfied almost surely, the expected reward will be $1$, while it will be $0$ otherwise.
When this property is expressed in LTL, the complexity increases to 2EXPTIME-hard \cite{Courco95}. 
Using the standard translation from LTL to NBAs and UCAs (e.g., \cite{Somenz00,BabiakKRS12}), the complexity becomes EXPTIME-hard for the former.

\begin{theorem}[Lower bounds]\label{theo:hard}
Finding an $\varepsilon$-optimal strategy for a lookahead decision process $\Mm_a = (S, s_0, A, T, r, \Aa_a)$ is EXPTIME-hard in the size of $\Aa_a$. If $\Aa_a$ is given as an LTL formula, the problem becomes 2EXPTIME hard.
\end{theorem}

\section{Removing Lookaheads}
\label{sec:translation}
The objective of this section is to establish a matching upper bound for Theorem~\ref{theo:hard}.
To meet the technical requirement of satisfying the $\omega$-regular promises, we will translate them to good-for-MDP automata \cite{Hahn20}. 
As the objectives are represented as universal co-B\"uchi automata, two operations are required: promise collection and translation to good-for-MDP NBAs. 
Promise collection is a simple operation for universal automata that does not impact the state space. 
However, translating an ordinary nondeterministic automaton to a good-for-MDP automaton, or even checking if an automaton has this property, can be a challenging task \cite{DBLP:journals/corr/abs-2202-07629}. 
Complementation alone is a costly operation \cite{stacs/Schewe09}.

We show that leading rank-based complementation procedures can be used to produce good-for-MDP (GFM) automata. 
Therefore, any standard implementation for automata complementation can be utilized. 
However, we suggest using a strongly limit-deterministic variant to avoid unnecessary nondeterminism, which is known~\cite{Hahn20} to affect the  efficiency of RL. 
Recall that an NBA is called \emph{limit deterministic} if it is deterministic after seeing the first final transition.
A limit deterministic automaton is \emph{strongly} limit deterministic if it is also deterministic \emph{before} taking the first final transition.
\begin{definition}
An automaton is \emph{strongly limit deterministic} if its state set $Q$ can be partitioned into sets $Q_1$ and $Q_2$, such that $|\delta(q,\sigma) \cap Q_1| \leq 1$ for all $q \in Q_1$ and $\sigma \in \Sigma$ and $|\delta(q,\sigma)| \leq 1$ and $\delta(q,\sigma) \subseteq Q_2$ for all $q \in Q_2$ and $\sigma \in \Sigma$, and the image of $\gamma$ is a subset of $Q_2$.
\end{definition}

Strongly limit deterministic NBAs are often good for MDPs, but they need not be (see Appendix~\ref{sec:strongLDBA}).

\subsection{From ordinary to collecting UCAs}
\label{ssec:collect}
We need to construct a GFM automaton that checks whether \underline{all} promises made on the future development of the MDP are almost surely fulfilled. The first step is to transform the given UCA schema for testing individual promises into a UCA that checks whether all promises are fulfilled.
When the promises are provided as states (or, indeed as sets of states) of a given UCA schema ${\mathcal A} = (\Sigma,Q,\delta,\gamma)$ and a fresh state $q_0' \notin Q$ and $Q' = Q \cup \{q_0'\}$, we define the \emph{collection automaton}
${\mathcal C} = (\Sigma \times Q,Q',q_0',\delta',\gamma')$, whose inputs $\Sigma \times 2^Q$ contains the ordinary input letter and a fresh promise,
\begin{itemize}
    \item $\delta'(q,(\sigma,q')) = \delta(q,\sigma)$ and $\gamma'(q,(\sigma,q')) = \gamma(q,\sigma)$ for all $q,q' \in Q$,
    that is, for states in $Q$, the promise is ignored, and
    
    \item $\gamma'(q_0',(\sigma,q)) = \delta(q,\sigma)$ and $\delta'(q_0',(\sigma,q)) = \{q_0'\}\cup\delta(q,\sigma)$,
    that is, from the fresh initial state $q_0'$, we have a non-final transition back to $q_0'$ as well as transitions that, broadly speaking, reflect the fresh promise $q \in Q$.
\end{itemize}

Note that promises can be restricted to be exactly or at most one state.
The reason that the transitions from $q_0'$ to other states are final is that this provides slightly smaller automata in the complementation (and determinisation) procedure we discuss in this section; as they can be taken only once on a run, it does not matter whether or not they are accepting, which can be exploited in a `nondeterministic determinisation procedure' as in \cite{Schewe09/det}.

Note that this automaton is easy to adjust to pledging acceptance from sets of states by using $\gamma'(q_0',(\sigma,S)) = \bigcup_{q\in S}\delta(q,\sigma)$ and $\delta'(q_0',(\sigma,S)) = \gamma'(q,(\sigma,S)) \cup \{q_0'\}$; the proofs in this section are easy to adjust to this case.

\begin{theorem}
\label{thm:collection}
For a given UCA schema $\mathcal A$, the automaton $\mathcal C$ from above accepts a word $\varpi = (\sigma_0,q_0)(\sigma_1,q_1)(\sigma_2,q_2)\ldots$ if, and only if, it satisfies all promises.
\end{theorem}

\subsection{From UCAs to (GFM) NBAs}
\label{ssec:complement}
Next, we consider a variation of the standard level ranking \cite{DBLP:journals/tocl/KupfermanV01,DBLP:journals/ijfcs/FriedgutKV06,stacs/Schewe09}, which is producing a semi-deterministic automaton.
This automaton is a syntactic subset in that it has the same states as \cite{stacs/Schewe09}, but only a subset of its transitions.
Besides being strongly limit-deterministic, we show that it retains the complement language and is good-for-MDPs.
Our construction follows the intuitive data structure from \cite{stacs/Schewe09}. 
It involves taking transitions away from the automaton resulting from the construction in \cite{stacs/Schewe09}, so that one side of the language inclusions is obtained for free, while the other side is entailed by the simulation presented in the Appendix~\ref{app:simulate}.

\paragraph{Construction.}
We call a level-ranking function $f :S {\to} \mathbb N$ from a finite set $S\subseteq Q$ of states $S$-tight if, for some $n \leq |S|$, it maps $S$ to $\{0,1,\ldots,2n{-}1\}$ and onto $\{1,3,\ldots,2n{-}1\}$. 
We write $\mathcal{T}_S$ for the set of $S$-tight level-ranking functions.
We call $\rank(f) = \max\{f(q) \mid q \in S\}$
(the $2n{-}1$ from above) the \emph{rank} of $f$.
\begin{definition}[Rank-Based Construction]
\label{def:comp}
    For a given $\omega$-automaton $\mathcal A=(\Sigma,Q,I,\delta,\gamma)$ with $n=|Q|$ states, let $\mathcal C=(\Sigma,Q',\{I\},\delta',\gamma')$ denote the NBA where
\begin{itemize}
\item $Q' = Q_1 \cup Q_2$ with $Q_1 = 2^Q$ and
$Q_2 = \{\,(S,O,f,i) \in 2^Q \times 2^Q \times \mathcal{T}_S \times \{0,2,\ldots,2n-2\} \mid O \subseteq f^{-1}(i)\,\}$,
\item $\delta' = \delta_1 \cup \delta_2 \cup \delta_3$ with
\begin{itemize}
\item $\delta_1: Q_1 \times \Sigma \rightarrow 2^{Q_1}$ with $\delta_1(S,\sigma)= \{\delta(S,\sigma)\}$,
\item $\delta_2: Q_1 \times \Sigma \rightarrow 2^{Q_2}$ with $(S',O,f,i)\in \delta_2(S,\sigma)$
iff $ S'{=} \delta(S,\sigma)$, $O{=}\emptyset$, and $i{=}0$,

\item 
$\delta_3: Q_2 \times \Sigma \rightarrow 2^{Q_2}$ with $(S',O',f',i')\in \delta_3\big((S,O,f,i),\sigma\big)$
 iff the following holds:
 
$S' = \delta(S,\sigma)$ and we define the auxiliary function
$g\colon S' \to 2^{\{0,\ldots,2n-1\}}$ with $g(q)$ equals
\[
\set{j \mid q \in \delta(f^{-1}(j),\sigma)} \cup \set{2\lfloor j/2 \rfloor \mid q \in \gamma(f^{-1}(j),\sigma)}
\]

$f'$ is the $S'$-tight function with $f'(q) = \min\{g(q)\}$; if this function is not $S'$-tight, the transition blocks. Otherwise:
\begin{enumerate}
    \item[(1)] we set $O'' = \delta(O,\sigma) \cap {f'}^{-1}(i)$
    \item[(2)] if $O'' \neq \emptyset$, then $O' = O''$ and $i'=i$;
    \item [(3)] else
$i' {=} (i {+} 2) \mod (\rank(f') {+}1)$ and $O' {=} {f'}^{-1}(i')$ \end{enumerate}
\end{itemize}
\item
$\gamma'$ contains the transitions of $\delta_3$ from case (3) (the breakpoints)
as well as transitions from  $\{\emptyset\}$.
\end{itemize}
\end{definition}

\newcommand{\randomstatsnumtopformulae}{5}
\newcommand{\randomstatsnumrandomlychosen}{5}
\newcommand{\randomstatscomplementmemorylimit}{4294967296}
\newcommand{\randomstatscomplementtimeout}{600}
\newcommand{\randomstatsnumatomicpropositions}{4}
\newcommand{\randomstatstotalentries}{10000}
\newcommand{\randomstatssuccessfulruns}{9947}
\newcommand{\randomstatssuccessfulrunspercent}{99.47\%}

\newcommand{\randomstatsmeanautomatonnumstates}{4.09}
\newcommand{\randomstatsstdevautomatonnumstates}{2.91}
\newcommand{\randomstatsminautomatonnumstates}{1}
\newcommand{\randomstatsmaxautomatonnumstates}{34}

\newcommand{\randomstatsmeancomplementnumstates}{48367.45}
\newcommand{\randomstatsstdevcomplementnumstates}{760963.32}
\newcommand{\randomstatsmincomplementnumstates}{2}
\newcommand{\randomstatsmaxcomplementnumstates}{25107909}

\newcommand{\randomstatsmeancomplementtime}{0.36}
\newcommand{\randomstatsstdevcomplementtime}{6.09}
\newcommand{\randomstatsmincomplementtime}{2.8133392333984375e-05}
\newcommand{\randomstatsmaxcomplementtime}{263.7465920448303}

\newcommand{\randomstatsmeancomplementprunerejectingnumstates}{3748.94}
\newcommand{\randomstatsstdevcomplementprunerejectingnumstates}{129045.67}
\newcommand{\randomstatsmincomplementprunerejectingnumstates}{1}
\newcommand{\randomstatsmaxcomplementprunerejectingnumstates}{9152588}

\newcommand{\randomstatsmeancomplementprunerejectingtime}{0.01}
\newcommand{\randomstatsstdevcomplementprunerejectingtime}{0.11}
\newcommand{\randomstatsmincomplementprunerejectingtime}{1.9073486328125e-05}
\newcommand{\randomstatsmaxcomplementprunerejectingtime}{4.555199861526489}

\newcommand{\randomstatsmeancomplementlumpdetpartnumstates}{23.80}
\newcommand{\randomstatsstdevcomplementlumpdetpartnumstates}{211.58}
\newcommand{\randomstatsmincomplementlumpdetpartnumstates}{1}
\newcommand{\randomstatsmaxcomplementlumpdetpartnumstates}{9958}

\newcommand{\randomstatsmeancomplementlumpdetparttime}{0.01}
\newcommand{\randomstatsstdevcomplementlumpdetparttime}{0.35}
\newcommand{\randomstatsmincomplementlumpdetparttime}{1.6927719116210938e-05}
\newcommand{\randomstatsmaxcomplementlumpdetparttime}{31.428340911865234}

\newcommand{\randomstatsmeancomplementredirectlanguagenumstates}{7.77}
\newcommand{\randomstatsstdevcomplementredirectlanguagenumstates}{10.39}
\newcommand{\randomstatsmincomplementredirectlanguagenumstates}{1}
\newcommand{\randomstatsmaxcomplementredirectlanguagenumstates}{327}

\newcommand{\randomstatsmeancomplementredirectlanguagetime}{0.01}
\newcommand{\randomstatsstdevcomplementredirectlanguagetime}{0.40}
\newcommand{\randomstatsmincomplementredirectlanguagetime}{3.2901763916015625e-05}
\newcommand{\randomstatsmaxcomplementredirectlanguagetime}{24.359985828399658}

\newcommand{\randomstatsmeancomplementlumpallnumstates}{7.03}
\newcommand{\randomstatsstdevcomplementlumpallnumstates}{8.70}
\newcommand{\randomstatsmincomplementlumpallnumstates}{1}
\newcommand{\randomstatsmaxcomplementlumpallnumstates}{320}

\newcommand{\randomstatsmeancomplementlumpalltime}{0.00}
\newcommand{\randomstatsstdevcomplementlumpalltime}{0.00}
\newcommand{\randomstatsmincomplementlumpalltime}{1.5020370483398438e-05}
\newcommand{\randomstatsmaxcomplementlumpalltime}{0.003160238265991211}

\newcommand{\randomstatsmeancomplementtotaltime}{0.40}
\newcommand{\randomstatsstdevcomplementtotaltime}{6.39}
\newcommand{\randomstatsmincomplementtotaltime}{0.001255035400390625}
\newcommand{\randomstatsmaxcomplementtotaltime}{269.6073257923126}

\begin{table*}
\caption{\label{tab:stats_table}Statistics for randomly generated examples.
 \textbf{orig}: Number of states of the automaton generated by \texttt{ltl2tgba},
 \textbf{compl}: Number of states of the complement, 
\textbf{prune}: Number of states after removing states with empty language,
\textbf{lumpd}: Number of states after applying strong-bisimulation lumping in the final part of the automaton,
\textbf{lang}: Number of states after we identify language-equivalent states in the final part and redirect transitions from the initial part to a representative for each language,
\textbf{lumpa}: Number of states after applying {\it strong bisimulation} lumping for all states of the automaton,
\textbf{time}: total time in seconds. 
}
\begin{center}
\begin{tabular}{lrrrrrrr}
\toprule
 & \multicolumn{1}{c}{\textbf{orig}} & \multicolumn{1}{c}{\textbf{compl}} & \multicolumn{1}{c}{\textbf{prune}} & \multicolumn{1}{c}{\textbf{lumpd}} & \multicolumn{1}{c}{\textbf{lang}} & \multicolumn{1}{c}{\textbf{lumpa}} & \multicolumn{1}{c}{\textbf{time}}\\
\midrule
\textbf{mean} & 4.09 & 48,367.45 & 3,748.94 & 23.80 & 7.77 & 7.03 & 0.40\\
\textbf{stdev} & 2.91 & 760,963.32 & 129,045.67 & 211.58 & 10.39 & 8.70 & 6.39\\
\textbf{max} & 34.00 & 25,107,909.00 & 9,152,588.00 & 9,958.00 & 327.00 & 320.00 & 269.61\\
\bottomrule
\end{tabular}

\end{center}
\end{table*}

\begin{table*}[t]
\caption{\label{tab:stats_examples_table}Example formulas.  For the legend, see Table~\ref{tab:stats_table}.
}
\begin{center}
\small

\begin{tabular}{p{7.5cm}rrrrrrr}
\toprule
 \multicolumn{1}{c}{\textbf{formula}} & \multicolumn{1}{c}{\textbf{orig}} & \multicolumn{1}{c}{\textbf{compl}} & \multicolumn{1}{c}{\textbf{prune}} & \multicolumn{1}{c}{\textbf{lumpd}} & \multicolumn{1}{c}{\textbf{lang}} & \multicolumn{1}{c}{\textbf{lumpa}} & \multicolumn{1}{c}{\textbf{time}}\\
\midrule
\texttt{Fd U ((a <-> Gd) \& (c <-> Fb))} & 14 & 25,107,909 & 16,585 & 2,120 & 115 & 60 & 269.61\\
\texttt{((c xor Fd) R F(b \& c)) W Xd} & 13 & 20,484,339 & 59,150 & 1,005 & 30 & 13 & 127.77\\
\texttt{F((a W (1 U (d xor Xd))) R (a W c))} & 10 & 19,317,020 & 18,540 & 103 & 40 & 29 & 167.43\\
\texttt{X(1 U a) R F(!Gb \& (c W a))} & 11 & 18,492,964 & 294,249 & 502 & 32 & 15 & 111.25\\
\texttt{G(Xa xor (G(Gc xor Ga) M Xd))} & 14 & 18,129,540 & 9,152,588 & 909 & 80 & 73 & 112.71\\
\midrule
\texttt{!G(a \& c) | X!Xa} & 2 & 4 & 2 & 2 & 2 & 2 & 0.00\\
\texttt{XG(Gd U (!a \& (c M Ga)))} & 1 & 2 & 2 & 2 & 2 & 2 & 0.00\\
\texttt{!(b M c) -> (c \& X!b)} & 3 & 6 & 3 & 3 & 3 & 3 & 0.00\\
\texttt{(Ga -> b) U c} & 4 & 8 & 6 & 6 & 6 & 6 & 0.01\\
\texttt{(!c R Fb) U (Gd <-> GFb)} & 10 & 232,094 & 70,513 & 6,481 & 60 & 15 & 11.76\\
\bottomrule
\end{tabular}

\end{center}
\end{table*}

\begin{theorem}
[\citet{stacs/Schewe09}] Given an NBA $\mathcal A$, the NBA $\mathcal C$ from Definition~\ref{def:comp} recognizes a subset of the complement of the language of $\mathcal A$. i.e. $\mathcal L(\mathcal C) \subseteq \Sigma^\omega \setminus \mathcal L(\mathcal A)$.
\end{theorem}

\begin{corollary}
\label{cor:complement} Given a UCA $\mathcal A$, the NBA $\mathcal C$ from Definition~\ref{def:comp} recognizes a subset of the language of  $\mathcal A$, i.e. $\mathcal L(\mathcal C) \subseteq \mathcal L(\mathcal A)$.
\end{corollary}

Showing inclusion in the other direction (and thus language equivalence) can be done in two ways.
One way is to re-visit the similar proof from the complementation construction from \cite{stacs/Schewe09}.
It revolves around guessing the correct level ranking once it is henceforth tight, and this guess, and its corresponding run, is still possible.
However, as we need to establish that the resulting NBA $\mathcal C$ is good-for-MDPs, we take a different approach: we start from determinising the UCA $\mathcal A$ into a deterministic Streett automaton $\mathcal S$, using the standard determinisation from nondeterministic B\"uchi to deterministic Rabin automata \cite{Schewe09/det} (provided in Appendix \ref{app:determinise}).
It is then easy to see how an accepting run of $\mathcal S$ on a word can be simulated.
The proof details are given in Appendix \ref{app:simulate}.
\begin{theorem}
\label{theo:gfm}
For a given UCA $\mathcal A$, the NBA $\mathcal C$ from Definition~\ref{def:comp} is a language equivalent good-for-MDPs NBA.\qed
\end{theorem}

Noting that the construction in Definition~\ref{def:comp} is a language equivalent syntactic subset of  \cite{stacs/Schewe09}, which in turn is a language equivalent syntactic subset for older constructions~\cite{DBLP:journals/tocl/KupfermanV01,DBLP:journals/ijfcs/FriedgutKV06,stacs/Schewe09}, we obtain that the classic rank-based complementation algorithms result in GFM automata.
\begin{corollary}\label{cor:classicComplement}
Given an NBA $\mathcal A$, the rank-based complementation algorithms from \cite{DBLP:journals/tocl/KupfermanV01,DBLP:journals/ijfcs/FriedgutKV06,stacs/Schewe09} provide good-for-MDP automata.
\qed
\end{corollary}

In Appendix \ref{app:opt}, we provide optimizations for this construction, showing in particular that (1) $\delta_2$ can be restricted to map all states to odd ranks and that (2) the state $q_0'$ from the collection automaton can always be chosen to be the sole state with maximal rank.
Further, we argue that safety and reachability objectives lead to subset and breakpoint constructions, respectively.

\subsection{Putting it all together}
Combining the selection of promises (Section \ref{ssec:collect}) and their efficient representation as a GFM automaton (Section \ref{ssec:complement}
), we can use them for model-checking and reinforcement learning \cite{Hahn20,Hahn19,Bozkur20,Hahn23}, including for hierarchical goals \cite{Bozkurt0P21,HahnPSSTW21}. 
\begin{theorem}\label{theo:complete}
The problem of finding (near) optimal control for a lookahead decision process $\Mm_a = (S, s_0, A, T, r, \Aa_a)$ can be done in time polynomial in $\Mm$ and is EXPTIME-complete in the size of $\Aa_a$, and 2EXPTIME-complete in the size of an LTL formula describing $\Aa_a$. \qed
\end{theorem}

\section{Experimental Results}
\label{sec:experimental}
Our experiments focus on showing that what could be a computational bottleneck (the size of resulting B\"uchi automaton) is not a showstopper. 
Once the automaton is produced, the scalability of our approach is similar to that of the lexicographic RL algorithm of~\cite{Hahn23}.

\vspace{0.4em}
\noindent\textbf{Efficiency of the Construction.} To obtain an estimate of the practical applicability of the complementation algorithm of Section~\ref{ssec:complement}, we implemented it and applied it to randomly generated formulas.
We generated a total of \randomstatstotalentries~random formulae using the \textsc{Spot} \cite{Duret-LutzLFMRX16} 2.11.3 tool \texttt{randltl} with \randomstatsnumatomicpropositions~atomic propositions each. We then converted each of these formulas to Büchi automata using \texttt{ltl2tgba}.
We used our prototypical tool to complement these automata with a timeout of \randomstatscomplementtimeout~seconds and were successful in \randomstatssuccessfulrunspercent~of the cases. We then applied several optimizations to reduce the number of states in the complement, all of which maintained the good-for-MDP property.
Table~\ref{tab:stats_table} provides statistics on our results, and Table~\ref{tab:stats_examples_table} provides individual values for some example runs. The first \randomstatsnumtopformulae\ entries are the ones for which the complementation led to the largest number of states, while the next \randomstatsnumrandomlychosen\ were randomly selected

As seen in Table~\ref{tab:stats_table}, the maximum number of complement states is more than a million, while the mean is much lower.
The standard deviation is quite high.
Looking at the data, this is because in most cases the number of states generated for the complement is relatively low, while in some cases it is very big.
As seen, all optimizations lead to a reduction, although the effect of applying bisimulation lumping to all states in the end is not as large as the other ones.
As seen in Table~\ref{tab:stats_examples_table}, in some cases the number of states was quite large.
However, after applying the optimizations described, we were able to further reduce the number of states to make the resulting automaton suitable for model checking or reinforcement learning.

\noindent \textbf{Case Study.} Our construction effectively reduces the optimization and RL problem for ODPs to lexicographic optimization/RL over MDPs. 
We combined our construction with the lexicographic $\omega$-regular and discounted objectives RL algorithm introduced in~\cite{Hahn23} to compute optimal policies shown in Figure~\ref{fig:biolab}.
It took $20$ mins on Intel $i7-8750H$ processor.

\section{Conclusion}
\label{sec:conclusion}
Successful reinforcement learning and optimal control often rely on the design of a suitable reward signal. While it's easy to design a reward signal as a function of the state and action for simpler problems, practical problems require non-Markovian rewards. Reward machines, formal specifications, and regular decision processes are some of the approaches used in this context.
We have introduced omega-regular decision processes (ODPs) as a formalism that provides great flexibility in specifying complex, non-Markovian rewards derived from a combination of qualitative and quantitative objectives. A key aspect of our approach is the ability for the decision maker to obtain rewards contingent upon the fulfillment of \emph{promises} in the language of expressive $\omega$-regular specifications.

Our algorithm reduces the ODP optimization problem to a lexicographic optimization problem over MDPs with $\omega$-regular and discounted reward objectives. This reduction is based on translating the collection semantics of promises to a good-for-MDPs B\"uchi automaton, which enables an automata-theoretic approach to optimization.
Our prototype implementation demonstrate the feasibility of our approach.

\bibliography{papers}

\begin{thebibliography}{41}
\providecommand{\natexlab}[1]{#1}

\bibitem[{Abadi and Brafman(2021)}]{Abadi21}
Abadi, E.; and Brafman, R.~I. 2021.
\newblock Learning and Solving Regular Decision Processes.
\newblock In \emph{Proceedings of the Twenty-Ninth International Joint
  Conference on Artificial Intelligence}.

\bibitem[{Amodei et~al.(2016)Amodei, Olah, Steinhardt, Christiano, Schulman,
  and Man{\'e}}]{amodei2016concrete}
Amodei, D.; Olah, C.; Steinhardt, J.; Christiano, P.; Schulman, J.; and
  Man{\'e}, D. 2016.
\newblock Concrete problems in AI safety.
\newblock \emph{arXiv preprint arXiv:1606.06565}.

\bibitem[{Babiak et~al.(2012)Babiak, K\v{r}et{\'{\i}}nsk{\'{y}}, Reh{\'{a}}k,
  and Strejcek}]{BabiakKRS12}
Babiak, T.; K\v{r}et{\'{\i}}nsk{\'{y}}, M.; Reh{\'{a}}k, V.; and Strejcek, J.
  2012.
\newblock {LTL} to {B\"uchi} Automata Translation: Fast and More Deterministic.
\newblock In \emph{Tools and Algorithms for the Construction and Analysis of
  Systems}, 95--109.

\bibitem[{Baier and Katoen(2008)}]{Baier08}
Baier, C.; and Katoen, J.-P. 2008.
\newblock \emph{Principles of Model Checking}.
\newblock MIT Press.

\bibitem[{Bozkurt, Wang, and Pajic(2021)}]{Bozkurt0P21}
Bozkurt, A.~K.; Wang, Y.; and Pajic, M. 2021.
\newblock Model-Free Learning of Safe yet Effective Controllers.
\newblock In \emph{2021 60th {IEEE} Conference on Decision and Control (CDC),
  Austin, TX, USA, December 14-17, 2021}, 6560--6565. {IEEE}.

\bibitem[{Bozkurt et~al.(2020)Bozkurt, Wang, Zavlanos, and Pajic}]{Bozkur20}
Bozkurt, A.~K.; Wang, Y.; Zavlanos, M.~M.; and Pajic, M. 2020.
\newblock Control Synthesis from Linear Temporal Logic Specifications using
  Model-Free Reinforcement Learning.
\newblock In \emph{International Conference on Robotics and Automation (ICRA)},
  10349--10355.

\bibitem[{Brafman and De~Giacomo(2019)}]{Brafm19}
Brafman, R.~I.; and De~Giacomo, G. 2019.
\newblock Planning for {LTLf /LDLf} Goals in Non-{Markovian} Fully Observable
  Nondeterministic Domains.
\newblock In \emph{Proceedings of the 28th International Joint Conference on
  Artificial Intelligence}, 1602–1608.

\bibitem[{Camacho et~al.(2019)Camacho, Icarte, Klassen, Valenzano, and
  McIlraith}]{camacho2019ltl}
Camacho, A.; Icarte, R.~T.; Klassen, T.~Q.; Valenzano, R.~A.; and McIlraith,
  S.~A. 2019.
\newblock {LTL} and Beyond: Formal Languages for Reward Function Specification
  in Reinforcement Learning.
\newblock In \emph{IJCAI}, volume~19, 6065--6073.

\bibitem[{Camacho and McIlraith(2019)}]{camacho2019strong}
Camacho, A.; and McIlraith, S.~A. 2019.
\newblock Strong Fully Observable Non-Deterministic Planning with LTL and LTLf
  Goals.
\newblock In \emph{IJCAI}, 5523--5531.

\bibitem[{Camacho et~al.(2018)Camacho, Muise, Baier, and
  McIlraith}]{camacho2018ltl}
Camacho, A.; Muise, C.~J.; Baier, J.~A.; and McIlraith, S.~A. 2018.
\newblock LTL Realizability via Safety and Reachability Games.
\newblock In \emph{IJCAI}, 4683--4691.

\bibitem[{Clark and Amodei(2016)}]{Coastrunners}
Clark, J.; and Amodei, D. 2016.
\newblock {Faulty Reward Functions in the Wild}.
\newblock \url{https://openai.com/blog/faulty-reward-functions/}.
\newblock Accessed on: 01/18/2023.

\bibitem[{Courcoubetis and Yannakakis(1995)}]{Courco95}
Courcoubetis, C.; and Yannakakis, M. 1995.
\newblock The Complexity of Probabilistic Verification.
\newblock \emph{J. ACM}, 42(4): 857--907.

\bibitem[{Dimitrova, Ghasemi, and Topcu(2018)}]{dimitrova2018maximum}
Dimitrova, R.; Ghasemi, M.; and Topcu, U. 2018.
\newblock Maximum realizability for linear temporal logic specifications.
\newblock In \emph{Automated Technology for Verification and Analysis: 16th
  International Symposium, ATVA 2018, Los Angeles, CA, USA, October 7-10, 2018,
  Proceedings 16}, 458--475. Springer.

\bibitem[{Duret{-}Lutz et~al.(2016)Duret{-}Lutz, Lewkowicz, Fauchille, Michaud,
  Renault, and Xu}]{Duret-LutzLFMRX16}
Duret{-}Lutz, A.; Lewkowicz, A.; Fauchille, A.; Michaud, T.; Renault, E.; and
  Xu, L. 2016.
\newblock Spot 2.0 - {A} Framework for {LTL} and $\omega$-Automata
  Manipulation.
\newblock In \emph{Automated Technology for Verification and Analysis},
  122--129.

\bibitem[{Filiot, Jin, and Raskin(2009)}]{filiot2009antichain}
Filiot, E.; Jin, N.; and Raskin, J.-F. 2009.
\newblock An antichain algorithm for LTL realizability.
\newblock In \emph{Computer Aided Verification: 21st International Conference,
  CAV 2009, Grenoble, France, June 26-July 2, 2009. Proceedings 21}, 263--277.
  Springer.

\bibitem[{Finkbeiner and Schewe(2013)}]{finkbeiner2013bounded}
Finkbeiner, B.; and Schewe, S. 2013.
\newblock Bounded synthesis.
\newblock \emph{International Journal on Software Tools for Technology
  Transfer}, 15(5-6): 519--539.

\bibitem[{Friedgut, Kupferman, and
  Vardi(2006)}]{DBLP:journals/ijfcs/FriedgutKV06}
Friedgut, E.; Kupferman, O.; and Vardi, M.~Y. 2006.
\newblock B{\"{u}}chi Complementation Made Tighter.
\newblock \emph{Int. J. Found. Comput. Sci.}, 17(4): 851--868.

\bibitem[{Fu and Topcu(2014)}]{Fu14}
Fu, J.; and Topcu, U. 2014.
\newblock Probably Approximately Correct {MDP} Learning and Control With
  Temporal Logic Constraints.
\newblock In \emph{Robotics: Science and Systems}.

\bibitem[{Goodfellow, Bengio, and Courville(2016)}]{Goodfe16}
Goodfellow, I.; Bengio, Y.; and Courville, A. 2016.
\newblock \emph{Deep Learning}.
\newblock MIT Press.

\bibitem[{Hahn et~al.(2019)Hahn, Perez, Schewe, Somenzi, Trivedi, and
  Wojtczak}]{Hahn19}
Hahn, E.~M.; Perez, M.; Schewe, S.; Somenzi, F.; Trivedi, A.; and Wojtczak, D.
  2019.
\newblock Omega-Regular Objectives in Model-Free Reinforcement Learning.
\newblock In \emph{Tools and Algorithms for the Construction and Analysis of
  Systems}, 395--412.
\newblock {LNCS} 11427.

\bibitem[{Hahn et~al.(2020)Hahn, Perez, Schewe, Somenzi, Trivedi, and
  Wojtczak}]{Hahn20}
Hahn, E.~M.; Perez, M.; Schewe, S.; Somenzi, F.; Trivedi, A.; and Wojtczak, D.
  2020.
\newblock Good-for-{MDPs} Automata for Probabilistic Analysis and Reinforcement
  Learning.
\newblock In \emph{Tools and Algorithms for the Construction and Analysis of
  Systems}, 306--323.
\newblock {LNCS} 12078.

\bibitem[{Hahn et~al.(2021)Hahn, Perez, Schewe, Somenzi, Trivedi, and
  Wojtczak}]{HahnPSSTW21}
Hahn, E.~M.; Perez, M.; Schewe, S.; Somenzi, F.; Trivedi, A.; and Wojtczak, D.
  2021.
\newblock Model-Free Reinforcement Learning for Lexicographic Omega-Regular
  Objectives.
\newblock In Huisman, M.; Pasareanu, C.~S.; and Zhan, N., eds., \emph{Formal
  Methods - 24th International Symposium, {FM} 2021, Virtual Event, November
  20-26, 2021, Proceedings}, volume 13047 of \emph{Lecture Notes in Computer
  Science}, 142--159. Springer.

\bibitem[{Hahn et~al.(2023)Hahn, Perez, Schewe, Somenzi, Trivedi, and
  Wojtczak}]{Hahn23}
Hahn, E.~M.; Perez, M.; Schewe, S.; Somenzi, F.; Trivedi, A.; and Wojtczak, D.
  2023.
\newblock Omega-Regular Reward Machines.
\newblock In \emph{{ECAI} 2023 - 26th European Conference on Artificial
  Intelligence, September 30 - October 4, 2023, Krak{\'{o}}w, Poland -
  Including 12th Conference on Prestigious Applications of Intelligent Systems
  {(PAIS} 2023)}, volume 372 of \emph{Frontiers in Artificial Intelligence and
  Applications}, 972--979. {IOS} Press.

\bibitem[{Icarte et~al.(2018)Icarte, Klassen, Valenzano, and McIlraith}]{RM1}
Icarte, R.~T.; Klassen, T.; Valenzano, R.; and McIlraith, S. 2018.
\newblock Using reward machines for high-level task specification and
  decomposition in reinforcement learning.
\newblock In \emph{International Conference on Machine Learning}, 2107--2116.

\bibitem[{Icarte et~al.(2022)Icarte, Klassen, Valenzano, and
  McIlraith}]{icarte2022reward}
Icarte, R.~T.; Klassen, T.~Q.; Valenzano, R.; and McIlraith, S.~A. 2022.
\newblock Reward machines: Exploiting reward function structure in
  reinforcement learning.
\newblock \emph{Journal of Artificial Intelligence Research}, 73: 173--208.

\bibitem[{Kupferman and Vardi(2001)}]{DBLP:journals/tocl/KupfermanV01}
Kupferman, O.; and Vardi, M.~Y. 2001.
\newblock Weak alternating automata are not that weak.
\newblock \emph{{ACM} Trans. Comput. Log.}, 2(3): 408--429.

\bibitem[{McDaniel and Einstein(2007)}]{mcdaniel2007prospective}
McDaniel, M.~A.; and Einstein, G.~O. 2007.
\newblock \emph{Prospective memory: An overview and synthesis of an emerging
  field}.
\newblock Sage Publications.

\bibitem[{Mirhoseini et~al.(2020)Mirhoseini, Goldie, Yazgan, Jiang, Songhori,
  Wang, Lee, Johnson, Pathak, Bae et~al.}]{mirhoseini2020chip}
Mirhoseini, A.; Goldie, A.; Yazgan, M.; Jiang, J.; Songhori, E.; Wang, S.; Lee,
  Y.-J.; Johnson, E.; Pathak, O.; Bae, S.; et~al. 2020.
\newblock Chip placement with deep reinforcement learning.
\newblock \emph{arXiv preprint arXiv:2004.10746}.

\bibitem[{Pan, Bhatia, and Steinhardt(2022)}]{pan2022effects}
Pan, A.; Bhatia, K.; and Steinhardt, J. 2022.
\newblock The effects of reward misspecification: Mapping and mitigating
  misaligned models.
\newblock \emph{arXiv preprint arXiv:2201.03544}.

\bibitem[{Pnueli(1977)}]{Pnueli77}
Pnueli, A. 1977.
\newblock The Temporal Logic of Programs.
\newblock In \emph{IEEE Symposium on Foundations of Computer Science}, 46--57.

\bibitem[{Puterman(1994)}]{Put94}
Puterman, M.~L. 1994.
\newblock \emph{Markov Decision Processes: Discrete Stochastic Dynamic
  Programming}.
\newblock New York, NY, USA: John Wiley \& Sons.

\bibitem[{Sadigh et~al.(2014)Sadigh, Kim, Coogan, Sastry, and
  Seshia}]{Sadigh14}
Sadigh, D.; Kim, E.; Coogan, S.; Sastry, S.~S.; and Seshia, S.~A. 2014.
\newblock A Learning Based Approach to Control Synthesis of {Markov} Decision
  Processes for Linear Temporal Logic Specifications.
\newblock In \emph{Conference on Decision and Control (CDC)}, 1091--1096.

\bibitem[{Schewe(2009{\natexlab{a}})}]{stacs/Schewe09}
Schewe, S. 2009{\natexlab{a}}.
\newblock B{\"{u}}chi Complementation Made Tight.
\newblock In Albers, S.; and Marion, J., eds., \emph{26th International
  Symposium on Theoretical Aspects of Computer Science, {STACS} 2009, February
  26-28, 2009, Freiburg, Germany, Proceedings}, volume~3 of \emph{LIPIcs},
  661--672. Schloss Dagstuhl - Leibniz-Zentrum f{\"{u}}r Informatik, Germany.

\bibitem[{Schewe(2009{\natexlab{b}})}]{Schewe09/det}
Schewe, S. 2009{\natexlab{b}}.
\newblock Tighter Bounds for the Determinisation of B{\"{u}}chi Automata.
\newblock In de~Alfaro, L., ed., \emph{Foundations of Software Science and
  Computational Structures, 12th International Conference, {FOSSACS} 2009, Held
  as Part of the Joint European Conferences on Theory and Practice of Software,
  {ETAPS} 2009, York, UK, March 22-29, 2009. Proceedings}, volume 5504 of
  \emph{Lecture Notes in Computer Science}, 167--181. Springer.

\bibitem[{Schewe, Tang, and
  Zhanabekova(2022)}]{DBLP:journals/corr/abs-2202-07629}
Schewe, S.; Tang, Q.; and Zhanabekova, T. 2022.
\newblock Deciding What is Good-for-MDPs.
\newblock \emph{CoRR}, abs/2202.07629.

\bibitem[{Silver et~al.(2016)Silver, Huang, Maddison, Guez, Sifre, Van
  Den~Driessche, Schrittwieser, Antonoglou, Panneershelvam, Lanctot
  et~al.}]{AlphaGo}
Silver, D.; Huang, A.; Maddison, C.~J.; Guez, A.; Sifre, L.; Van Den~Driessche,
  G.; Schrittwieser, J.; Antonoglou, I.; Panneershelvam, V.; Lanctot, M.;
  et~al. 2016.
\newblock Mastering the game of Go with deep neural networks and tree search.
\newblock \emph{nature}, 529(7587): 484--489.

\bibitem[{Skalse et~al.(2022)Skalse, Howe, Krasheninnikov, and
  Krueger}]{skalse2022defining}
Skalse, J.; Howe, N.~H.; Krasheninnikov, D.; and Krueger, D. 2022.
\newblock Defining and characterizing reward hacking.
\newblock \emph{arXiv preprint arXiv:2209.13085}.

\bibitem[{Somenzi and Bloem(2000)}]{Somenz00}
Somenzi, F.; and Bloem, R. 2000.
\newblock Efficient {B\"uchi} Automata from {LTL} Formulae.
\newblock In \emph{Computer Aided Verification}, 248--263.
\newblock {LNCS} 1855.

\bibitem[{Sutton and Barto(2018)}]{Sutton18}
Sutton, R.~S.; and Barto, A.~G. 2018.
\newblock \emph{Reinforcement Learning: An Introduction}.
\newblock MIT Press, second edition.

\bibitem[{Wurman et~al.(2022)Wurman, Barrett, Kawamoto, MacGlashan,
  Subramanian, Walsh, Capobianco, Devlic, Eckert, Fuchs
  et~al.}]{wurman2022outracing}
Wurman, P.~R.; Barrett, S.; Kawamoto, K.; MacGlashan, J.; Subramanian, K.;
  Walsh, T.~J.; Capobianco, R.; Devlic, A.; Eckert, F.; Fuchs, F.; et~al. 2022.
\newblock Outracing champion Gran Turismo drivers with deep reinforcement
  learning.
\newblock \emph{Nature}, 602(7896): 223--228.

\bibitem[{Yuan et~al.(2019)Yuan, Yu, Gu, Deng, and Li}]{yuan2019novel}
Yuan, Y.; Yu, Z.~L.; Gu, Z.; Deng, X.; and Li, Y. 2019.
\newblock A novel multi-step reinforcement learning method for solving reward
  hacking.
\newblock \emph{Applied Intelligence}, 49: 2874--2888.

\end{thebibliography}

\appendix
\onecolumn

\section{Strongly Limit-Deterministic NBA}
\label{sec:strongLDBA}
\begin{proposition}
    \label{prop:strongLDBA}
    Strongly limit deterministic NBA may not be good-for-MDP.
\end{proposition}
\begin{proof}
The NBA shown in the figure below depicts a \emph{strongly limit deterministic} NBA that accepts every word, but is not good-for-MDP, because it needs to predict the next letter on the fly when moving to the second phase.
\begin{center}
\scalebox{0.8}{
\begin{tikzpicture}
    \node[state, initial, fill=safecellcolor] (S0) {$q_0$};
    \node[state, fill=safecellcolor] (S1) [above right=1cm and 2cm  of S0] {$q_1$};
    \node[state, fill=safecellcolor] (S2) [below right=1cm and 2cm of S0] {$q_2$};
    \node[state, fill=safecellcolor] (S3) [right=4cm of S0] {$q_3$};
    \path[->]
    (S0) edge [loop below] node {$\top$}(S0)
    (S3) edge [loop right] node[accepting dot,label={$\top$}]{}(S3)
    (S0) edge node {$\top$} (S1)
    (S0) edge[swap] node {$\top$} (S2)
    (S1) edge node {$p$} (S3)
    (S2) edge node[below right] {$\neg p$} (S3);
  \end{tikzpicture}
}
\end{center}    
\end{proof}

\section{Proof of Theorem~\ref{thm:collection}}
\label{sec:collection}
\begin{proof}
If $\varphi$ does not satisfy all promises, there is an $i\in \omega$ such that $w_i=\sigma_{i}\sigma_{i+1} \ldots$ is not in the language of $\mathcal A_{q_i}$, which in turn implies that there is a rejecting run
$q_{i},\sigma_{i},q_{i+1}',\sigma_{i+1},q_{i+2}'$ run of $\mathcal A_{q_i}$ on $w_i$.
But then 
\begin{multline*}
    q_0',(\sigma_0,q_0),q_0', \ldots, q_0', (\sigma_{i},q_{i}),q_{i+1}', (\sigma_{i+1},q_{i+1}),q_{i+1}', (\sigma_{i+2},q_{i+1}),q_{i+2}',\ldots
\end{multline*} 
is a rejecting run of $\mathcal C$ on $\varpi$.
Vice versa, if $\varpi$ is not accepted by $\mathcal C$, then there is a rejecting run
\begin{multline*}
q_0',(\sigma_0,q_0),q_0', \ldots, q_0', (\sigma_{i},q_{i}),q_{i+1}', (\sigma_{i+1},q_{i+1}),q_{i+1}', (\sigma_{i+2},q_{i+1}),q_{i+2}',\ldots    
\end{multline*}
 of $\mathcal C$ on $\varpi$, which has $q_0'$ appearing $i{+}1$ times for some $i \in \omega$. (Note that $q_0'$, being the initial state of $\mathcal C$, appears at least once on all runs; it can only be reached from itself, and hence only appears in an initial segment of the run; a run where it appears all the time is not rejecting.)
But then
the word $w_i=\sigma_{i}\sigma_{i+1} \ldots$ is not in the language of $\mathcal A_{q_i}$, because the run $q_{i},\sigma_{i},q_{i+1}',\sigma_{i+1},q_{i+2}'$ of $\mathcal A_{q_i}$ on $w_i$ is rejecting.
\end{proof}

\section{From UCAs to DSAs}
\label{app:determinise}

Efficient determinization of NBAs to deterministic Rabin (and thus of UCAs to deterministic Streett) automata is built around \emph{history trees} \cite{Schewe09/det}. History trees are an abstraction of the possible initial sequences of runs of a B\"uchi automaton $\mathcal A$ on an input word $\alpha$.
Our construction is taken from there, and the same construction with an illustrative example can be found in \cite{Schewe09/det}.
The details of the construction are needed for Lemma~\ref{lemma:simulation}.

An \emph{ordered tree} $T \subseteq \omega^*$ is a finite prefix and order closed subset of finite sequences of natural numbers.
That is, if a sequence $\tau=t_0,t_1, \ldots t_n \in T$ is in $T$, then all sequences
$s_0,s_1, \ldots s_m$ with $m\leq n$ and, for all $i \leq m$, $s_i \leq t_i$, are also in $T$.
For a node $\tau \in T$ of an ordered tree $T$, we call the number of children of $\tau$ its \emph{degree}, denoted by $\deg_T(\tau)=|\{i \in \omega \mid \tau \cdot i \in T\}|$.

A \emph{history tree}
for a given NBA $\mathcal A= (\Sigma,Q,I,\delta,\gamma)$ is a labeled tree $\langle T,l \rangle$, where $T$ is an ordered tree, and $l:T \rightarrow 2^Q\smallsetminus \{\emptyset\}$ is a labeling function that maps the nodes of $T$ to non-empty subsets of $Q$, such that
1) the label of each node is a proper superset of the union of the labels of its children, and
2) the labels of different children of a node are disjoint.
We call a node $\tau$ the \emph{host node} of a state $q$, if $q\in l(\tau)$ is in the label of $\tau$, but not in the label of any child of $\tau$.

\subsection{History Transitions}

For a given nondeterministic B\"uchi automaton $\mathcal A= (\Sigma,Q,I,\delta,\gamma)$, history tree $\langle T,l \rangle$, and input letter $\sigma \in \Sigma$, we construct the
\emph{$\sigma$-successor} $\langle \widehat{T},\widehat{l}\rangle$ of $\langle T,l \rangle$ in four steps.
In a first step
we construct the labeled tree $\langle T',l': T' \rightarrow 2^Q \rangle$ such that
\begin{itemize}
\item $\tau \in T' \supset T$ is a node of $T'$ if, and only if, $\tau \in T$ is in $T$ or $\tau =\tau' \cdot \deg_T(\tau')$ is formed by appending the degree $\deg_T(\tau')$ of a node $\tau' \in T$ in $T$ to $\tau'$,

\item the label $l'(\tau) = \delta(l(\tau),\sigma)$ of an old node $\tau \in T$ is the set $\delta(l(\tau),\sigma)=\bigcup_{q \in l(\tau)}\delta(q,\sigma)$ of $\sigma$-successors of the states in the label of $\tau$, and

\item the label $l'(\tau \cdot \deg_T(\tau')) = \gamma(l(\tau),\sigma)$ of a new node $\tau \cdot \deg_T(\tau)$ is the set of \emph{final} $\sigma$-successors of the states in the label of $\tau$.
\end{itemize}

After this step, each old node is labeled with the $\sigma$-successors of the states in its old label, and every old node $\tau$ has spawned a new child $\tau'=\tau\cdot\deg(\tau)$, which is labeled with the states reachable through accepting transitions.

The new tree is not necessarily a history tree: (1) nodes may be labeled with an empty set, (2) the labels of siblings do not need to be disjoint, and (3) the union of the children's labels do not need to form a proper subset of their parent's label.

In the second step, property (2) is re-established:
we construct the tree $\langle T',l'': T' \rightarrow 2^Q \rangle$, where $l''$ is inferred from $l'$ by removing all states in the label of a node $\tau'= \tau \cdot i$ and all its descendants if it appears in the label $l'(\tau \cdot j)$ of an older sibling ($j<i$).

Properties (1) and (3) are re-established in the third transformation step.
In this step, we construct the tree $\langle T'',l'': T'' \rightarrow 2^Q \rangle$ by (a) removing all nodes $\tau$ with an empty label $l''(\tau)=\emptyset$, and (b) removing all descendants of nodes whose label is partitioned by the labels of its children from $T'$.
(We use $l''$ in spite of the type mismatch, strictly speaking we should use its restriction to $T''$.)
We call the greatest prefix and order closed subset of $T''$ the set of \emph{stable} nodes and the stable nodes whose descendants have been deleted due to rule (b) \emph{collapsing}.

The tree resulting from this transformation satisfies the properties (1)--(3), but it is no longer order closed.
In order to obtain a proper history tree, the order closedness is re-established in the final step of the transformation.
We construct the \mbox{$\sigma$-successor} $\langle \widehat{T},\widehat{l}: \widehat{T} \rightarrow 2^Q \smallsetminus \{\emptyset\}\rangle$ of $\langle T,l \rangle$ by ``compressing'' $T''$ to a an order closed tree, using the compression function $\comp:T'' \rightarrow \omega^*$ that maps the empty word $\varepsilon$ to $\varepsilon$, and $\tau \cdot i$ to $\comp(\tau) \cdot j$, where $j= |\{k < i \mid \tau \cdot k \in T''\}|$ is the number of older siblings of $\tau \cdot i$.
For this function $\comp:T'' \rightarrow \omega^*$, we simply set $\widehat{T} = \{ \comp(\tau) \mid \tau \in T''\}$ and $\widehat{l}(\comp(\tau)) = l''(\tau)$ for all $\tau \in T''$.
The nodes that are renamed during this step are exactly those that are not stable.

\subsection{Deterministic Acceptance Mechanism}
For an NBA $\mathcal A= (\Sigma,Q,I,\delta,\gamma)$, we call the history tree $\langle T_0,l_0 \rangle = \langle \{\varepsilon\}, \varepsilon \mapsto I \rangle$ that contains only the empty word and maps it to the initial states $I$ of $\mathcal A$ the initial history tree.

For an input word $\alpha:\omega \rightarrow \Sigma$
we call the sequence $\langle T_0,l_0 \rangle, \langle T_1,l_1 \rangle, \ldots$ of history trees that start with the initial history tree $\langle T_0,l_0 \rangle$ and where, for every $i\in \omega$, $\langle T_i,l_i \rangle$ is followed by $\alpha(i)$-successor $\langle T_{i+1},l_{i+1} \rangle$ the \emph{history trace} or $\alpha$.
A node $\tau$ in the history tree $\langle T_{i+1},l_{i+1} \rangle$ is called stable or collapsing, respectively, if it is stable or collapsing in the $\alpha(i)$-transition from $\langle T_i,l_i \rangle$ to $\langle T_{i+1},l_{i+1} \rangle$.

\begin{theorem}\cite{Schewe09/det}
\label{prop:accept}
An $\omega$-word $\alpha$ is accepted by a nondeterministic B\"uchi automaton $\mathcal A$ if, and only if, there is a node $\tau \in \omega^*$ such that $\tau$ is eventually always stable and always eventually collapsing in the history trace of $\alpha$.
\end{theorem}

\begin{corollary}
\label{cor:determinise}
An $\omega$-word $\alpha$ is accepted by a universal co-B\"uchi automaton $\mathcal A$ if, and only if, it is accepted by the determinstic Streett automaton $\mathcal S$ described earlier in this appendix that checks that, for every node $\tau \in \omega^*$,
$\tau$ is only finitely often collapsing or infinitely often unstable in the history trace of $\alpha$.
\end{corollary}

\section{Establishing Good-for-MDPness}
\label{app:simulate}
As a final step, we establish that the strongly limit deterministic NBA $\mathcal C$ from Section \ref{ssec:complement} is (1) good-for-MDPs and (2) language equivalent to the DSA $\mathcal S$ constructed in
Appendix \ref{app:determinise} by showing that it can simulate $\mathcal S$ with some extra information.

This simulation game is played between a verifier, who wants to prove GFM-ness of an automaton $\mathcal C$, and a spoiler, who wants to disprove it and plays on an automaton $\mathcal M$, which is GFM.
It was first suggested in \cite{Hahn20} and works as follows:
Spoiler and verifier start in the respective states of their automata and take turns in selecting transitions (letters and successor states), such that, in every round
\begin{enumerate}
    \item the spoiler chooses first and
    \item the verifier then follows by choosing a transition with the same letter. 
\end{enumerate}
Moreover, the spoiler will have to declare at some point of the game, which transitions she will choose infinitely often, and that she will not choose any other transition henceforth, such that runs that follow this declaration are accepting for $\mathcal M$.
The spoiler wins if
\begin{enumerate}
    \item she has made this declaration,
    \item she has followed this declaration, and
    \item the verifier does not construct an accepting run.
\end{enumerate}

Otherwise, the verifier wins.

\begin{lemma}\label{lemma:simulation}
The verifier wins this simulation game when he plays on the NBA $\mathcal C$ constructed from an UCA $\mathcal A$ (Section \ref{ssec:complement}),  while the spoiler plays with the Streett automaton $\mathcal S$ from 
Appendix \ref{app:determinise}.
\end{lemma}

\begin{proof}
To show this, we describe the winning strategy.
Before the spoiler declares, the verifier simply uses the subset construction from the first phase.
If the spoiler never declares, or declares a set of transitions such that some node is henceforth stable and infinitely often accepting, then the verifier wins automatically.
We now look at the case where the spoiler declares, and all nodes that are henceforth stable are never again accepting.

Say we have $k$ nodes that are henceforth stable. Note that they must build an order and prefixed closed tree $T_s$.
We order them to $\tau_0,\ldots,\tau_{k-1}$, such that every node has a higher prefix than all of its descendants, and every node has a higher index than its older siblings and their descendants. (There is exactly one such order. The root, for example, has index $k-1$ in it.)

If the reachable states of $\mathcal A$ are $S$ and the spoiler has moved to a state $\langle T,l \rangle$, then obviously $T_s \subseteq T$, and the verifier can move (using the same letter) to a state that uses the $S$-tight function $f$ determined as follows.

Let $i$ be the smallest index such that $q \in l(\tau_i)$. (Note that such an index exists, as $q$ is a reachable state and $q\in l(\tau_{k-1})$ must therefore hold.)

If $q$ is also in a label $l(\tau')$ for a child $\tau'$ of $q$ (which then must be a child that is not henceforth stable, as the henceforth stable children of $\tau_i$ have a smaller index), then we assign $f(q) \mapsto 2i$; otherwise we assign $f(q)=2i+1$. 

It is now easy to see that, so long as the nodes of $T_s$ are stable, this relationship between node labels for the states of $\mathcal S$ and ranks of $\mathcal C$ is maintained.

To see that the run of $\mathcal C$ the spoiler produces is accepting, we assume for contradiction that, with the $j$'s transition (after the declaration of the spoiler, where the verifier always traverses an accepting transition), the verifier's run has seen the last accepting transition, and that the states in the run henceforth are \begin{multline*}
    (S_{j},O_{j},f_{j},2i'),(S_{j+1},O_{j+1},f_{j+1},2i'),
(S_{j+2},O_{j+2},f_{j+2},2i'),\ldots,
\end{multline*} 
while the states of the run constructed by the spoiler in the same positions are $\langle T_{j},l_{j}\rangle\langle T_{j+1},l_{j+1}\rangle,\langle T_{j+2},l_{j+2}\rangle,\ldots$.

We first observe that $O_j$ contains the union of the states in the labels $l(\tau')$ of all children of $\tau_{i'}$ in $T_j \setminus T_s$ (the children of $\tau_i$ not declared henceforth stable in the declaration of the spoiler).
By a similar inductive argument as before, $O_{j'}$, for $j' \geq j$, contains the states in the label $l(\tau')$ of those children of $\tau_{i'}$ that are in $T_{j'} \setminus T_s$ \emph{and} have not been spawned no later than in step $j$ of the run.
Yet, this in particular means that the oldest child of $\tau'$ not in $T_s$ has been removed at most as many times as $\mathcal A$ has states (as $\tau'$ cannot have more then this number of children).
But if this is the case, this child of $\tau'$ is not marked as unstable infinitely many times as promised by the spoiler, so that the spoiler loses the game.
\end{proof}

This simulation lemma shows that the language of $\mathcal S$ is included in the language of $\mathcal C$ ($\mathcal L(\mathcal S) \subseteq \mathcal L(\mathcal C )$).
Together with Corollaries \ref{cor:complement} and \ref{cor:determinise}, this implies that $\mathcal A$, $\mathcal C$, and $\mathcal S$ are language equivalent.
As $\mathcal S$ is deterministic, it is also good-for-MDPs, and with this language equivalence, the simulation theorem from \cite{Hahn20} provides that $\mathcal C$ is good-for-MDPs.

\begin{corollary}
\label{cor:gfm}
For a given UCA $\mathcal A$, the automaton $\mathcal C$ constructed in Section \ref{ssec:complement} is a language equivalent good-for-MDPs NBA.\qed
\end{corollary}

\section{Optimizations and Special Cases}\label{app:opt}
In this appendix, we provide two independent optimizations on the size of the automaton and the consideration of safety and reachability as special cases.

\begin{enumerate}
    \item The first optimization is on the size of the statespace and refers particularly to the collection automaton from Section \ref{ssec:collect}.
When we look at the simulation lemma from
Appendix \ref{app:simulate}, we see that the states (or rather: the state) with the highest rank are those, who occur only in the root node -- and for the collection automaton, it is easy to see that this is $q_0'$, and $q_0'$ only.
We can therefore restrict the statespace to states, where the tight level rankings assign $q_0'$, and $q_0'$ only, the highest rank.

\item 
The second improvement works for all source automata, and restricts the transitions that leave the subset-construction part:
It is enough to use the transitions that map all states to odd ranks.
This is so because an automaton that uses only these transition to traverse from the first to the second deterministic part can simply simulate the transition from $S_{j}$ to $(S_{j+1},O_{j+1},f_{j+1},i_{j+1})$
by a transition to $(S_{j+1},O_{j+1}',g_{j+1},i_{j+1}')$, such that $e(q) = 2\lfloor f(q)/2\rfloor +1$ holds.
This results in runs
$S_0,\sigma_0,\ldots,S_j,\sigma_j,(S_{j+1},O_{j+1},f_{j+1},i_{j+1}),\ldots$ and 
$S_0,\sigma_0,\ldots,S_j,\sigma_j,(S_{j+1},O_{j+1}',e_{j+1},i_{j+1}'),\ldots$, where a simple inductive argument provides that the following holds for all $j' \geq j$ and all potential even ranks~$i$:
\begin{itemize}
    \item $e^{-1}(i) \supseteq f^{-1}(i)$ and
    \item $e^{-1}(i) \cup e^{-1}(i-1) = f^{-1}(i) \cup f^{-1}(i-1)$, and therefore also
    \item $e^{-1}(i-1) \subseteq f^{-1}(i-1)$.
\end{itemize}

The latter entails that, when the latter run is in index $i'$ and the former enters it, then the latter run will reach the next breakpoint not later than the former. Thus, if the former is accepting, so is the latter.

This leaves the two special cases of safety and reachability automata.
Reachability UCAs are UCAs with only final transitions, so that the collecting automaton has only final transitions, except for the self-loop on $q_0'$.
The GFM NBA resulting from it only needs ranks $0$ and $1$, with $q_0'$ being the only state with rank $1$. It is essentially a breakpoint construction, which checks that, except for the run that stays in $q_0'$, all runs are finite.

Safety UCAs are UCAs with a rejecting sink (a state with a final self-loop), while all other transitions are non-final.
If we consider the GFM NBA resulting from the collecting automaton to a safety UCA, then we can see that all states that contain this rejecting sink are non-productive and can be removed.
The resulting GFM NBA then only needs three ranks, $1$ through $3$, where only $q_0'$ has rank $3$, while all other reachable states have rank $1$.
We can, however, simply adjust the collection automaton for this case by making \emph{all} transitions from $q_0'$ non-final. 
This changes at most one transition on each run from final to non-final, and thus has no bearing on acceptance.

Using the same construction for this adjusted collection automaton will still result in a GFM NBA, where all states that contain the sink are non-productive and can be removed. It then needs only a single rank, $1$, and will have only accepting transitions in the second part. Moreover, we observe that the first part is just a non-accepting copy of the second, and we have essentially a subset construction (unsurprising for a safety property).
\end{enumerate}

\end{document}